\documentclass[10pt,journal,compsoc]{IEEEtran}

\ifCLASSOPTIONcompsoc
  \usepackage[nocompress]{cite}
\else
  \usepackage{cite}
\fi                            

\usepackage{amsmath}
\usepackage{amsthm}
\usepackage{bm}
\usepackage{enumerate}
\usepackage{mathtools}
\usepackage{amsfonts}
\usepackage{algorithmicx}
\usepackage{algpseudocode}

\usepackage{booktabs}
\usepackage{multirow}
\usepackage{graphicx}
\usepackage{xcolor}

\usepackage{subcaption}
\usepackage[pagebackref,breaklinks,colorlinks]{hyperref}

\usepackage{balance}

\usepackage{utfsym}

\newtheorem{definition}{Definition}

\newtheorem{proposition}{Proposition}

\usepackage[utf8]{inputenc} 
\usepackage[T1]{fontenc}    
\usepackage{hyperref}       
\usepackage{url}            
\usepackage{booktabs}       
\usepackage{amsfonts}       
\usepackage{nicefrac}       
\usepackage{microtype}      

\usepackage{bm}

\usepackage{colortbl}

\usepackage{enumerate}

\usepackage[short]{optidef}

\usepackage[linesnumbered,ruled,vlined]{algorithm2e}
\SetKwInput{KwInput}{Input}                \SetKwInput{KwOutput}{Output} 

\SetCommentSty{mycommfont}

\usepackage[misc]{ifsym}

\usepackage{subcaption}
\begin{document}
\title{
RECESS Vaccine for Federated Learning: Proactive Defense Against Model Poisoning Attacks
}
\author{
Haonan~Yan,
Wenjing~Zhang,
Qian~Chen,
Xiaoguang~Li,
Wenhai~Sun,
Hui~Li,
and~Xiaodong~Lin
\IEEEcompsocitemizethanks{
\IEEEcompsocthanksitem 
Haonan Yan, Qian Chen, Xiaoguang Li, and Hui Li are with the Xidian University, China.
\IEEEcompsocthanksitem 
Wenjing Zhang and Xiaodong Lin are with the University of Guelph, Canada.
\IEEEcompsocthanksitem 
Wenhai Sun is with the Purdue University, USA.
\IEEEcompsocthanksitem 
Xiaoguang Li and Hui Li are the corresponding authors. 

}
}

\IEEEtitleabstractindextext{%
\begin{abstract}
Model poisoning attacks greatly jeopardize the application of federated learning (FL).
The effectiveness of existing defenses is susceptible to the latest model poisoning attacks, 
leading to a decrease in prediction accuracy.
Besides, these defenses are intractable to distinguish benign outliers from malicious gradients, which further compromises the model generalization.
In this work, we propose a novel proactive defense named {\sf RECESS} against model poisoning attacks.
Different from the passive analysis in previous defenses, {\sf RECESS} proactively queries each participating client with a delicately constructed aggregation gradient, accompanied by the detection of malicious clients according to their responses with higher accuracy.
Furthermore, RECESS uses a new trust scoring mechanism to robustly aggregate gradients. Unlike previous methods that score each iteration, RECESS considers clients' performance correlation across multiple iterations to estimate the trust score, substantially increasing fault tolerance.
%
%
%
Finally, we extensively evaluate {\sf RECESS} on typical model architectures and four datasets under various settings.
We also evaluated the defensive effectiveness against other types of poisoning attacks, the sensitivity of hyperparameters, and adaptive adversarial attacks.
Experimental results show the superiority of {\sf RECESS} in terms of reducing accuracy loss caused by the latest model poisoning attacks over five classic and two state-of-the-art defenses.

\end{abstract}


}

\maketitle 

\IEEEdisplaynontitleabstractindextext
\IEEEpeerreviewmaketitle

\IEEEraisesectionheading{\section{Introduction} \label{Introduction}}

\IEEEPARstart{R}{ecently},
federated learning (FL) goes viral 
as a privacy-preserving training solution with the distributed learning paradigm \cite{li2020federated}, since data privacy 
attracts increasing
attention from organizations like banks \cite{yang2019ffd} and hospitals \cite{xu2021federated}, governments like GDPR \cite{truong2021privacy} and CPPA \cite{lackey2021data}, and commercial companies like Google \cite{mcmahan2017communication}. 
FL allows data owners to collaboratively train models under the coordination of a central server for better prediction performance by sharing local gradient updates instead of their own private/proprietary datasets, 
preserving the privacy of each participant's raw data.
FL is promising as a trending privacy training technology. 

However, FL is vulnerable to various model poisoning attacks \cite{guerraoui2018hidden,bhagoji2019analyzing}.
Due to its distributed characteristics, the local training datasets are unverifiable, so the attacker can upload malicious model updates to corrupt the aggregation global model,
which greatly limit the application of FL. 
Accordingly, the attacker can corrupt the global model by hijacking compromised participants, and then uploading malicious local gradient updates, leading to a reduction in the final model's prediction performance, which significantly 
hinders the flourishing of FL.

To mitigate model poisoning attacks, 
many byzantine-robust aggregation rules,
e.g., Krum, Mkrum, Trmean, Median, and Bulyan \cite{blanchard2017machine,yin2018byzantine,guerraoui2018hidden}, 
are proposed to remove malicious gradients by statistical analyses.
Two state-of-the-art (SOTA) defenses, FLTrust \cite{Cao2021FLTrustBF} and DnC \cite{Shejwalkar2021ManipulatingTB}, are also proposed to further enhance the robustness of the FL system.
Even so, these defenses 
are susceptible to the
latest optimization-based model poisoning attacks
\cite{fang2020local,Shejwalkar2021ManipulatingTB}, leading to a decrease in prediction accuracy.

Besides, 
there is an irreversible generalization loss in 
these defenses since they cannot distinguish malicious gradients and benign outliers, reducing the generalization performance of the final trained model.
Due to the unbalanced and not independently and identically distributed (Non-IID) local datasets across participating clients, 
it is common to have many biased updates which are statistical outliers.
These statistical outliers are still benign updates (called \textit{benign outliers}), which are helpful in improving the generalization performance of the final FL model. 
However, existing defenses typically rely on statistical majorities to remove malicious clients, which usually misclassifies the benign outliers as malicious clients and remove important gradient contributions from outliers directly, thereby compromising the generalization of FL.
To the best of our knowledge, the work of tackling benign outliers still remains open in FL.

To improve model accuracy and generalization against latest model poisoning attacks in FL, we propose {\sf RECESS} to p{\sf R}oactively d{\sf E}te{\sf C}t mod{\sf E}l poi{\sf S}oning attack{\sf S}.
Unlike previous defenses using passive analysis and direct aggregation in a single iteration,
{\sf RECESS} \textit{proactively detects} malicious clients and \textit{robustly aggregates} gradients with a new trust scoring based mechanism.
The differences 
among previous defenses and our {\sf RECESS} 
are illustrated in Fig. \ref{fig:intro}. 
\begin{figure*}
	\centering
	\includegraphics[width=\linewidth]{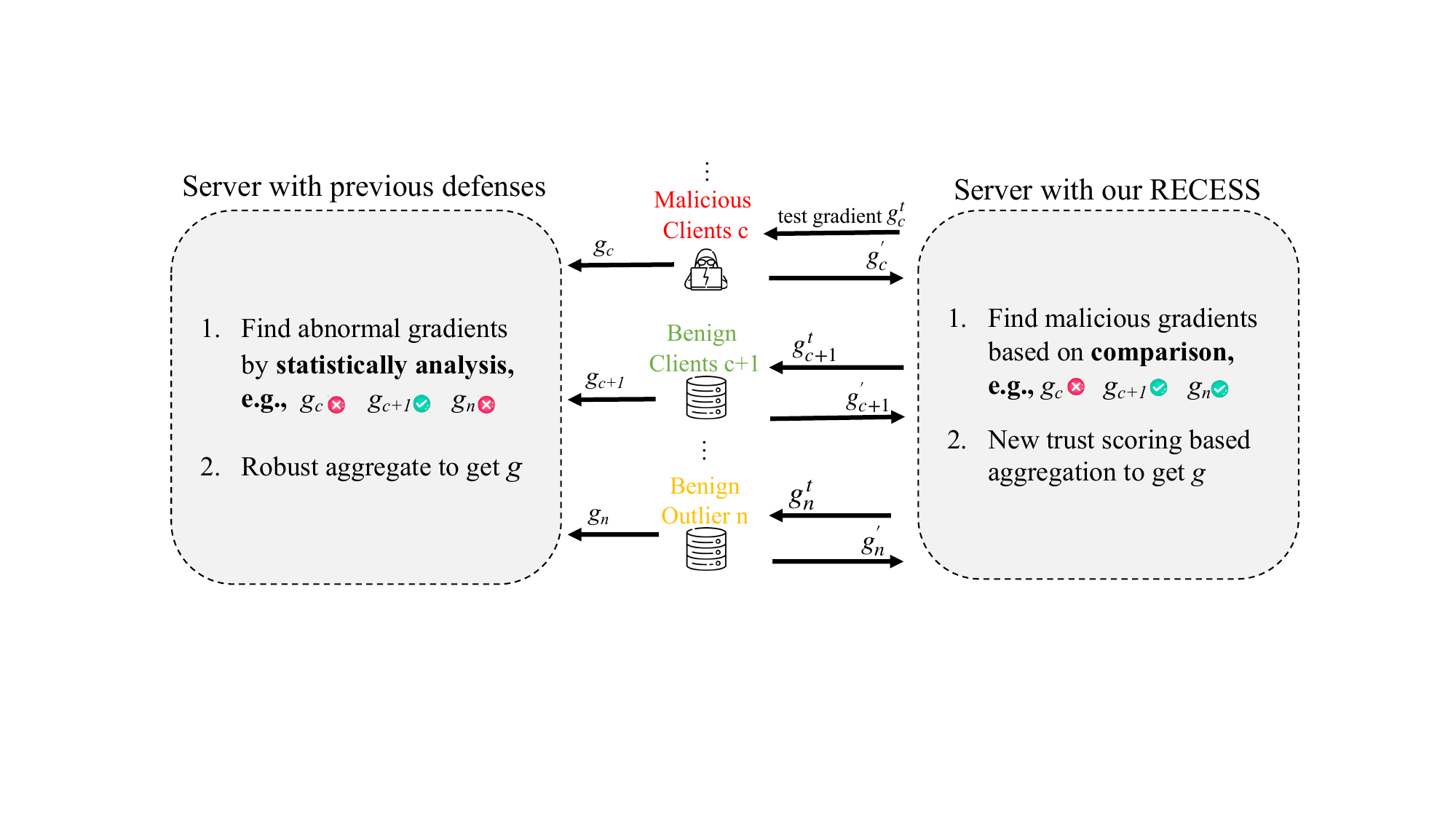}
	\caption{
	An intuitive example to illustrate the difference between previous defenses and {\sf RECESS}. In an FL system with $n$ clients, suppose the first $c$ clients are malicious.
	Previous defenses find the malicious gradients by abnormality detection of all received gradients, while {\sf RECESS} confirms the malicious gradients by comparing the constructed test gradient with the corresponding client's response.
	Thus, {\sf RECESS} can distinguish between malicious clients and benign outliers, 
	achieving higher accuracy in the final trained model.
	}
	\label{fig:intro}
\end{figure*}

In \textit{proactive detection}, our key insight is that the update goals of malicious clients are different from benign clients and outliers. 
Benign clients, including outliers, always tend to optimize the received aggregation gradient towards the distribution of their local datasets, thus uploading directive and stable gradient updates in the whole process.
However, malicious clients aim to deviate the aggregation result to a diverging direction of the aggregation gradient, resulting in inconsistent and sharply changing uploaded gradients.
Therefore, the defender can amplify this difference by delicately constructing aggregation gradients sent to clients, and then more accurately detect malicious clients and distinguish benign outliers by the comparison between constructed gradients and the corresponding responses. 

In \textit{robust aggregation}, 
we propose a new trust scoring based mechanism 
to aggregation gradients,
which estimates the score according to the user's performance over multiple iterations, rather than scoring each iteration as in previous methods, which improves the detection fault tolerance.

Finally, we compare {\sf RECESS} with five classic and two SOTA defenses under various settings,
and the results demonstrate that {\sf RECESS} is effective to 
overcome the two limitations aforementioned.

The main contributions are:
\begin{enumerate}[(a)]
	\item We propose a novel defense against model poisoning attacks called {\sf RECESS}. It offers a new defending angle for FL and turns the defender from passive analysis to proactive detection, which defeats the latest model poisoning attacks.
	{\sf RECESS} can also identify benign outliers effectively.
	\item We improve the robust aggregation mechanism. 
	A new trust scoring method is devised by considering
	clients' abnormality over multiple iterations, which significantly increases the detection fault tolerance of {\sf RECESS} and further improves the model's accuracy.
	\item We evaluate {\sf RECESS} under various settings.
	Comprehensive experimental results show that {\sf RECESS} outperforms previous defenses in terms of accuracy loss and achieves consistent effectiveness. 
\end{enumerate}

\section{Related Works}

\subsection{Federated Learning}


\textbf{FedAvg.} 
FedAvg \cite{mcmahan2017communication} is a commonly used aggregation rule for FL.
The aggregation gradient is a weighted average of each client's upload gradient,
and the weight is determined by the number of training data. 
However, the aggregation gradient, i.e., the global model, is vulnerable to poisoning attacks, 
resulting in a global model with poor prediction performance
\cite{blanchard2017machine}.

\subsection{Poisoning Attacks in FL}
From the perspective of the attacker's goal, 
poisoning attacks are categorized as targeted and untargeted attacks.
The targeted attack \cite{bagdasaryan2020backdoor,bhagoji2019analyzing} aims to 
mislead the global model to misclassify samples of one 
attacker-chosen class,
does not affect other classes.
The untargeted attack \cite{mahloujifar2019universal,guerraoui2018hidden,xie2020fall} aims to increase the testing error of the global model for all classes.
These poisoning goals can be accomplished by different operations, which are mainly classified into
two categories, data and model poisoning attacks.
Data poisoning attacks \cite{jagielski2018manipulating,munoz2017towards} contaminate the training dataset to mislead the training process of local models in FL.
For example, the label flipping attack altering the data's label to mislead the training of the model.
Model poisoning attacks \cite{fang2020local,Shejwalkar2021ManipulatingTB} directly poison the gradient updates to corrupt the server's aggregation.

In this work, we focus on the stronger untargeted model poisoning attacks for three reasons:
\begin{enumerate}[(a)]
\item Model poisoning attacks with the direct manipulation of gradients are more threatening than data poisoning attacks in FL \cite{bhagoji2019analyzing,fang2020local}.
\item Data poisoning attacks are considered a special case of model poisoning attacks as malicious gradients can be obtained on poisoning datasets \cite{Cao2021FLTrustBF}.
\item Untargeted poisoning is a more severe threat to model prediction performance than the targeted one in FL \cite{Shejwalkar2021ManipulatingTB}.
%
%
\end{enumerate}
In the following, we introduce three latest model poisoning attacks.


\textbf{LIE Attack.}
The LIE attack \cite{baruch2019little} lets the defender remove the non-byzantine clients and shift the aggregation gradient by carefully crafting byzantine values that deviate from the correct ones as far as possible.
This attack does not need to know the information of other benign clients.

\textbf{Optimization Attack.}
\cite{fang2020local} propose a new attack idea. They formulate the model poisoning attack as an optimization program, where the objective function is to let the poisoning aggregation gradient be far from the aggregation gradient without poisoning.
The solution is the crafted malicious gradients.
However, they obtain the solution by halving the search without optimization, so their attack effect is not optimal.
Though their solution is not optimal, this optimization attack achieved the best attack impact at that time. 

\textbf{AGR Attack Series.}
\cite{Shejwalkar2021ManipulatingTB} then improve the optimization program by introducing perturbation vectors and scaling factors.
They put the perturbation vector into the optimization objective, which is able to repeatedly fine-tune the malicious gradient for maximum attack effect. 
They propose three specific attack instances (AGR-tailored, AGR-agnostic Min-Max, and Min-Sum) .
For the AGR-tailored attack, when the attacker knows the servers' aggregation rules, they can tailor the objective function to the known rules and maximize the perturbation.
When the attacker does not know the rules, the objective is to let the poisoning gradient be far away from the majority of benign updates and also circumvent the defender's detection.
In the AGR-agnostic Min-Max attack, the malicious gradient's maximum distance from any other gradient is less than the distance between any two benign gradients.
In the AGR-agnostic Min-Sum attack, the distance is replaced by the sum of squared distances, which is more concealed.
They are considered the strongest poisoning attacks now.
They also consider two scenarios in which the attacker has or dose not have the knowledge of the aggregation rules and propose corresponding attacks.


\textbf{MPAF Attack.}
The above attacks are for compromised genuine clients. 
Recent attack called MPAF \cite{cao2022mpaf} uses fake clients.
MPAF injects fake clients rather than compromising victim clients, allowing a higher ratio of malicious users. In addition, MPAF uses a very simple attack method that merely attempts to drag the global model towards an attacker-chosen low-accuracy model, with a fixed local update goal and without optimization during poisoning. RECESS does not consider this type of poisoning attack, mainly because:
\begin{enumerate}[(a)]
	\item The attacker's assumptions do not conform to real-world FL requirements. We have concerns about its practical feasibility due to the strong assumption of arbitrary fake client injection required. To the best of our knowledge, satisfying such assumptions is highly challenging in practice, as evidenced by the following literature \cite{shejwalkar2022back}.
	\item Even if these assumptions held, basic Byzantine-robust aggregation rules could readily defeat the attack. Our proposal can be enhanced by using a common Median or Krum instead of weighted averaging in RECESS to defend against MPAF.
	\item Although MPAF submits malicious gradients, it is essentially a data poisoning attack before the FL system starts. Existing detections during training, like FLDetector, also do not consider such attacks.
\end{enumerate}

Given our reservations on the practicality of this attack, we are worried that elaborating on it in depth might give readers the false impression that it poses a serious threat. 
This paper focuses on defending against attacks under more realistic assumptions that are standard in the field.


\subsection{Existing Defenses}
\subsubsection{Byzantine-robust Aggregation Rules.}\label{sec:bagg}
We first review five classic byzantine-robust aggregation rules, which are regarded as the SOAT passive defenses.

\textbf{Krum.}
\cite{blanchard2017machine} propose 
a majority-based approach called Krum.
Suppose there are $n$ clients and $c$ malicious clients, Krum selects one gradient that minimizes the sum of squared distances to $n-c-2$ neighbors as the final aggregation result.

\textbf{Mkrum.}	
\cite{blanchard2017machine} propose a variant of Krum called Mkrum.
Mkrum iteratively uses the Krum function $m$ times to select the set of gradients without put-back, and the final aggregation gradient is the mean of the selection set.
Note that Mkrum is Krum when $m =1 $, and Mkrum is FedAvg when $m=n$.

\textbf{Trmean.}
 \cite{yin2018byzantine} improve the FedAvg and propose coordinate-wise Trmean.
Trmean removes the smallest and largest $k$ gradients and averages the remaining gradient as the aggregation result. 

\textbf{Median.}
\cite{yin2018byzantine} propose another coordinate-wise aggregation rule called Median.
This rule uses the median value of each dimension value of all upload gradients as the final output.

\textbf{Bulyan.}
 \cite{guerraoui2018hidden} combine the Krum and Trmean into a new robust aggregation rule called Bulyan.
Bulyan first selects several gradients with Krum and then takes the trimmed mean value of the selected set as the final result.

\subsubsection{Other Defense.}
\textbf{Gradient Clipping.}
 The gradient clipping strategy directly normalizes the gradient, which severely slows down the training, while {\sf RECESS} does not modify the gradient itself, 
 we only consider the magnitude as one of the factors to measure abnormality. 
 Furthermore, clipping is a one-size-fits-all approach that does not consider benign outliers, reducing the model generalization, and cannot effectively resist the latest poisoning attacks, while {\sf RECESS} addresses these issues.

\subsubsection{SOTA Defenses.}\label{sec:ed}
There are two directions to defend against the model poisoning attack: robust aggregation and anomaly detection.
Five classic byzantine-robust aggregation rules and FLTrust 
belong to the robust aggregation. 
DnC is one of the representative anomaly detectors.
Here we mainly describe two SOTA defenses.

\textbf{FLTrust.}
\cite{Cao2021FLTrustBF} involves the server with a small dataset to participate in each iteration and generate a gradient benchmark in each iteration. Updates far from the benchmark will be reduced in aggregation weights. 
Moreover, FLtrust normalizes the gradient magnitude to limit the impact of malicious gradients and achieves a competitive accuracy.
However, a small dataset is less representative, especially in Non-IID cases it is insufficient to represent all benign outliers where small datasets are less representative.
Thus, FLTrust would discard benign outliers.
In fact, using a static benchmark (e.g., median, FLTrust) to detect the malicious client will always implicate benign outliers.
Nevertheless, {\sf RECESS} improves this limitation through proactive detection and multi-round evaluation. 

\textbf{DnC.}	
\cite{Shejwalkar2021ManipulatingTB} leverages the spectral method to detect malicious gradients, which is proven effective in centered learning.
They also use random sampling to reduce memory and computational cost.
Note we will evaluate these two defense mechanisms, which have rigorous theoretical guarantees, as baselines.

\section{Threat Model}
We present our threat model,
, including the goal, capability, and knowledge of both the attacker and defender.
To be fair to compare with previous works \cite{Cao2021FLTrustBF, Shejwalkar2021ManipulatingTB}, we consider various knowledge of attackers comprehensively and commence the work under the same threat model.
We do not weaken the attacker with any additional disadvantage.

\subsection{Attacker}
\subsubsection{Attacker's Goal.}
We consider the untargeted model poisoning attack, where the attacker aims to increase the testing error rate of the global FL model.

\subsubsection{Attacker's Capability.}

The attacker controls several clients in the FL to poison the aggregation. 
In each iteration, the attacker has the full training information about controlled malicious clients.
The gradient updates of controlled clients can be manipulated by the attacker.
The attacker is also able to manipulate the gradient updates of the controlled clients and cannot control the server directly. 
In addition, the attacker is assumed to not compromise the server.

\subsubsection{Attacker's Knowledge.}
We consider 
the knowledge of other benign gradients and the server's aggregation rule in FL.
The attacker knows other benign gradients in the white-box setting, but not in the black-box setting.
Knowing the aggregation rule depends on the design of the attack.
In this work, the optimization attack and the AGR-tailored attack need to know the aggregation rules in their design, but the LIE attack, AGR-agnostic Min-Max, and Min-Sum attacks do not need to know.

We evaluate {\sf RECESS} and other defenses under the white/black-box setting.
The attacker may also know uploaded gradients of other benign clients to construct more powerful poisoning attacks.
Similarly, if the attacker knows the aggregation rule, the poisoning attacks would be more threatening.
The pure white-box scenario is also considered where the attacker knows the existence of {\sf RECESS} and how it works.

\subsection{Defender}
\subsubsection{Defender's Goals.}
The defender has two goals:
\begin{enumerate}[(1)]
\item Resist model poisoning attack.
The defender should find the malicious clients accurately, and minimize the accuracy loss caused by model poisoning attacks.
\item Guarantee model's training accuracy.
The defender should improve the aggregation for competitive accuracy, especially when there is no attack (known as fidelity).
Besides, 
the defender should \textbf{not} misclassify benign outliers,
which is not considered in previous works. 
\end{enumerate}
Previous defenses can not satisfy the second goal since they are based on statistical analysis, and {\sf RECESS} outperforms them in both goals.

\subsubsection{Defender's Capability.}
The defender can exercise rights on behalf of the server, including accepting gradients from clients, 
selecting arbitrary gradients to aggregate,
and sending the result back to arbitrary clients in each iteration.
We also consider a more practical setting where the defender can examine the uploads of clients without 
maintaining a clean root dataset on the server side like \cite{Cao2021FLTrustBF}, where the root dataset must be close to the local training datasets.

\subsubsection{Defender's Knowledge.}
The defender knows the aggregation rule.
But, the defender does not know the training dataset of local clients and the number of malicious clients.

\section{Our defense}
We propose a new proactive defense against model poisoning attacks in FL called {\sf RECESS}.
We first take an overview of {\sf RECESS}.
Then we introduce the proactive detection and the new trust scoring based aggregation mechanism respectively, followed by the extension of the new defined metric and effectiveness analysis. 

\subsection{Overview}
\subsubsection{Intuition.}\label{sec:intui}
{\sf RECESS} aims to distinguish malicious and benign gradients without implicating benign outliers.
The essential difference 
between benign and malicious clients 
is that:
benign clients including outliers always optimize the received aggregation result to the distribution of their local dataset to minimize the loss value,
while malicious clients aim to make the poisoned aggregation direction as far away from the aggregation gradient without poisoning as possible to maximize the poisoning effect.
In other words, benign gradients explicitly point to their local distribution, which is directional and relatively more stable, while the malicious gradients solved based on other benign gradients change inconsistently and dramatically.
The workflow of \textsf{RECESS} is shown in Fig. \ref{framework}.

\begin{figure}
	\centering
	\includegraphics[width=\linewidth]{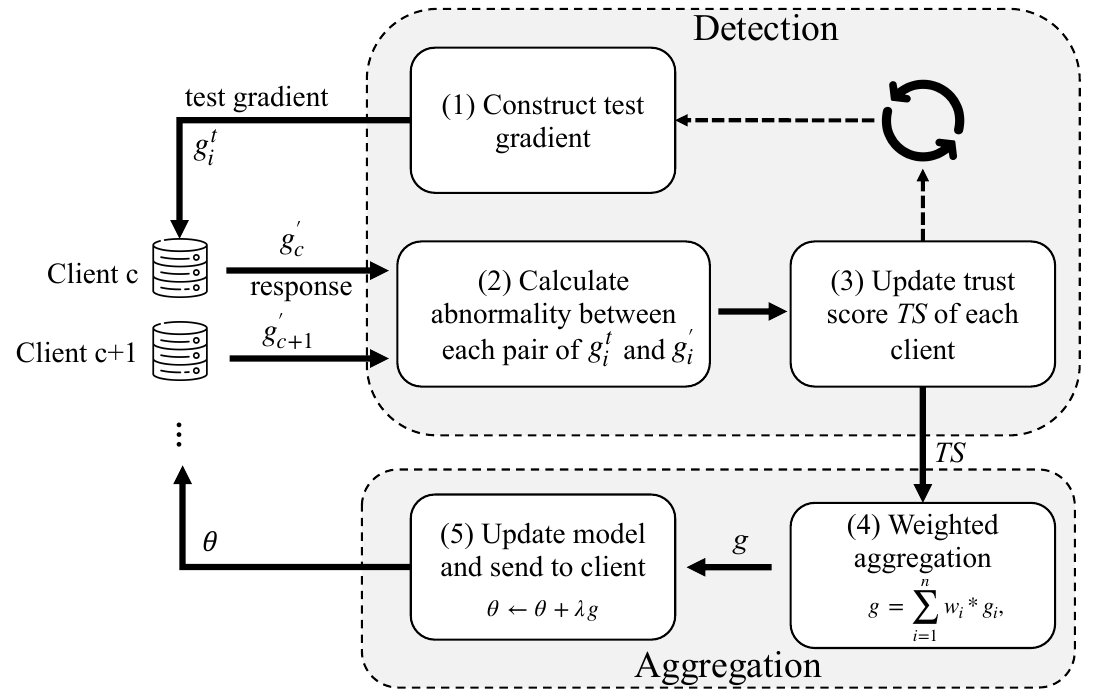}
	\caption{The overall detection process of \textsf{RECESS} against model poisoning attacks.}
	\label{framework}
\end{figure}

\subsubsection{Proactive Detection.} 
{\sf RECESS} is different from previous passive analyses.
The defender using {\sf RECESS} proactively tests each client with elaborately constructed aggregation gradients, rather than the model weights to clients, and these two approaches are equivalent algorithmically. 
After clients fine-tune local models with received aggregation gradients and respond with new gradient updates, 
the defender can recognize malicious clients based on the abnormality of their gradient updates.
We also redefine the abnormality in poisoning scenarios to distinguish malicious gradients and outliers, which promotes the generalization of the FL model.
Details are shown in Section \ref{sec:Alg}. 

\subsubsection{Robust Aggregation.} 
{\sf RECESS} adopts a new trust scoring based aggregation mechanism to improve the defense.
A new scoring method is proposed to give weights to aggregate clients' gradients.
Updates with greater abnormal extent account for a smaller proportion in aggregation, which increases fault tolerance.
Details are shown in Section \ref{sec:scheme}.

\subsection{Proactive Detection Algorithm}\label{sec:Alg}
We firstly show the construction strategy of the test gradient sent to clients, and then 
detect the abnormality of gradient from two dimensions.

\subsubsection{Construction Strategy of Test Gradient.}
The purpose of the defender is to observe the client's response by adding a slight perturbation to the test gradient. 
Thus, any slight varying gradient is feasible.
To illustrate, we present a strategy based on previous gradients.
In details, 
the defender firstly stores the aggregation gradient in the last iteration.
When entering the detection mode, for example in $(k-1)$-th round of detection, the defender idles this iteration, and for each uploaded gradient $\bm{g}_i^{k-1}$ from $i$-th client ($0< i \leq n$, $n$ is the number of clients), 
the defender scale down the magnitude of $\bm{g}_i^{k-1}$, i.e., $\|{\bm{g}_i^t}\|_2 = \bm{g}_i^{k-1}/ \|\bm{g}_i^{k-1}\|_2$, 
and slightly adjust the direction of 
${\bm{g}_i^t}$ by heuristically selecting several elements and adding noise.
Here we set a threshold $\mathcal{A}$ to control the direction adjustment.
We also discuss the setting of the threshold $\mathcal{A}$ in the experiments. 
Then, the defender feedbacks ${\bm{g}_i^t}$ as the aggregation gradient to client $i$. 
The client $i$ updates this tailored “aggregation gradient” locally and respond with a newly updated gradient $\bm{g}_i^{k}$.
The defender can perform poisoning detection by the comparison between ${\bm{g}_i^t}$ and $\bm{g}_i^{k}$.
This process will be repeated for each client.
Besides, to actively test malicious clients, the server with {\sf RECESS} deployed needs to send gradients rather than the model weight to clients, and these two approaches are equivalent algorithmically. 

{\sf RECESS} slightly adjusts constructed aggregation gradients, which magnifies the variance of malicious gradients and makes them more conspicuous,  
but benign clients and outliers are not affected.
The reason is:
for benign clients including outliers, their updated gradients always point to the direction of decreasing loss function values.
Although some machine learning optimization algorithms, such as SGD, have deviations, they are still unbiased estimations of the normal gradient descent direction as a whole \cite{bottou2010large}.
Conversely, malicious clients' gradients are
usually opposite to the aggregation direction and 
come from the solution of the attack optimization program on other benign gradients.
Thus, the variance of malicious gradients is enlarged by the optimization project.
Many works \cite{zhao2018federated,zhu2021federated,blanchard2017machine,bhagoji2019analyzing} also indicate that the fluctuation of the upload gradients cause the aggregation gradient to change more. 
Even when the defender uses the poisoned aggregation gradient to construct test gradient, this basis still remains unchanged and malicious and benign clients behave quite differently in this test.
Besides, {\sf RECESS} without extra model training is more efficient than other model-based detectors.
Concluding, 
{\sf RECESS} is effective to detect model poisoning attacks.

\subsubsection{Poisoning Detection.}
After receiving clients' responses, the defender detects malicious clients from two dimensions of abnormality: \textit{direction} and \textit{magnitude} of gradient changes before and after the test.

{\it Direction.}
Attacker can manipulate the controlled clients' gradients in an arbitrary direction to deviate the global gradient from the original direction.
In {\sf RECESS}, when receiving the test global gradient, the gradient direction of each benign client usually points to the distribution of its local datasets. 
The direction of each malicious gradient usually changes.  
Thus, it is necessary to measure the direction changes. 
Formally, we use metric {\it cosine similarity} $S_C^k$ to measure the angular change in direction between the  constructed gradient $\bm{g}_i^t$ and the response $\bm{g}_i^{k}$,
\begin{equation}
	S_{C}^k(\bm{g}_i^t,\bm{g}_i^{k})=\frac{\bm{g}_i^t\cdot\bm{g}_i^{k}}{\|\bm{g}_i^t\|_2 \cdot \|\bm{g}_i^{k}\|_2}.
\end{equation}
FLTrust \cite{Cao2021FLTrustBF} has also considered the direction change, in which
clients' gradients far from the direction of the server's gradient will be discarded directly, 
but this sacrifices benign outliers and decreases the model's generalization.
Different from the root direction in \cite{Cao2021FLTrustBF}, we consider the range of direction changes before and after the test.
\cite{Cao2021FLTrustBF} must ensure the distribution of local datasets be close to the  of the root dataset, and abandons the local updates whose directions are opposite to the root direction, which would sacrifice the benign outliers and decrease the generalization of the final trained model.
Our design is more robust and efficient. 

{\it Magnitude.}
Besides direction, the magnitude of the malicious gradient also dominates the poisoning effect, especially when larger than the benign gradients.
Here we utilize the $l_2$ distance to measure the magnitude, i.e., $\|\bm{g}_i^k\|_2$. 

After that, we redefine the abnormality of the client gradient by the deviation extent in direction and magnitude, instead of 
the distance from other selected gradients (e.g. Krum, Mkrum, Median, server's gradient \cite{Cao2021FLTrustBF}) or some benchmarks (e.g. Trmean, Bylan).


\begin{definition}\label{def:abn}
\textbf{(Abnormality)}
In model poisoning attacks of FL, the abnormality $\alpha$ of $k$-th uploaded gradient from the $i$-th client should be measured by
\begin{equation}
	\alpha= 
	\begin{cases}
	-S_C^k \cdot\left\|g_i^k\right\|_2, & S_C^k<0 \\ 
	-\frac{S_C^k}{\left\|g_i^k\right\|_2}, & S_C^k \geq 0
	\end{cases}.	
\end{equation}
\end{definition}


The cosine similarity $S_C^k$ controls the positive and negative of the abnormality $\alpha$.
When the direction of the client's response gradient $\bm{g}_i^{k}$ is inconsistent with the defender's test gradient $\bm{g}_i^t$, $\alpha$ will be positive, and vice versa.
The amount of change in $\alpha$ is related to 
$S_C^k$ and $\|\bm{g}_i^k\|_2$. 
The smaller the deviation between the direction of $\bm{g}_i^t$ and $\bm{g}_i^{k}$, the smaller the $\alpha$.
Meanwhile, the larger the $\|\bm{g}_i^k\|_2$, the higher the $\alpha$.
This setting encourages small-gradient updates, which avoids the domination of aggregation by malicious updates usually with a larger magnitude.

This new \textit{Abnormality} can be used in both existing defense directions, robust aggregation and anomaly detection.
In this work, we focus on the more robust aggregation mechanism and leave the latter in Section \ref{sec:appdis}.

\subsection{New Robust Aggregation Mechanism}\label{sec:scheme}
After detection, we propose a new trust scoring based mechanism to aggregate gradients, which increases {\sf RECESS}'s fault tolerance for false detection.



\subsubsection{Trust Scoring.}
{\sf RECESS} estimates the trust score for the long-term performance of each client, rather than the single iteration considered by previous schemes. 

The detail of the scoring process is: 
the defender first sets the same initial trust score $TS_0$ for each client.
After each round of detection, a constant baseline score is deducted for clients who are detected as suspicious poisoners and not selected for aggregation.
When the trust score of one client $i$ reaches zero, i.e., $TS_i = 0$, the defender will label this client as malicious and no longer involve this client's updated gradients in the aggregation.
The $TS$ is calculated by
\begin{equation}
	TS = TS_0 - \alpha * baseline\_decreased\_score.
\end{equation}
The $TS_0$ and $baseline\_decreased\_score$ controls the detection speed.

To prevent the attacker from increasing the trust score by suspending poisoning or constructing benign gradients, 
we let the defender defer adding $TS$ to further restrict the attacker.
Only after a period of good performance (e.g.,  10 consecutive rounds), $TS$ increase, but the score would be deducted once the client is detected as malicious, which effectively
forces the attacker not to poison or reduce the intensity of poisoning to a negligible level.

\subsubsection{Aggregation.}
After assigning $TS$ to all clients,
we transform $TS$ into weights which are used to average clients' gradients as the aggregation gradient $g$, i.e.,
\begin{equation}
	g=\sum_{i=1}^n w_i * g_i,
	\label{gradi}
\end{equation}
the weight of client $i$ is calculated by
\begin{equation}
w_i = \frac{e^{{TS}_i}}{\sum_{j=1}^{n} e^{{TS}_j}},
\label{weight}
\end{equation}
where $i,j = 1, 2,..., n$.

Note that in practice the server may select a subset of clients to aggregate, our mechanism also supports this case where the aggregated gradient should be the weighted average from the selected clients.

\subsection{Algorithm Pseudocode of {\sf RECESS}}\label{appendix_alg}
We show the algorithm pseudocode of {\sf RECESS} in Algorithm \ref{alg:recess} including our proactive detection and trust score based robust aggregation.

\begin{algorithm}[!ht]
\DontPrintSemicolon  
\KwInput{
$G$: the set of all uploaded gradients;

$TS_0$: initial trust score;

$BS$: baseline decreased score.
}
\KwOutput{
$g$: the aggregation gradient.
}
   
%



$n$ = len($G$)\tcp*{number of clients}

$TS = [TS_0] * n$\tcp*{initial trust score for each client}

\If{detection mode is True}{

\tcc{construct test gradient}

$G_t$ = $G$\tcp*{List of constructed test query}

\For{$i$ in n }{

$g_i$= scale down ($G$[i])
 
$g_i'$= adjust direction ($g_i$ )

$G_t$[i] = $g_i'$
}


$G'$ = send to clients and receive updates ($G_t$)

\tcc{update the trust score of each client}

marker = [0] $*$ n \tcp*{marker of delayed score adding}

\For{i in n}{
$S_C$  = cosine similarity (${G_t}$[i], $G'$[i])\tcp*{direction}

$M = \| G' \text{[i]} \|_2$\tcp*{magnitude}

$\alpha = -S_{C}^k / M $\tcp*{abnormality value}


\tcc{update score when the marker is eligible}

{\bf if} {$\alpha \leq 0$} {\bf then} marker[i]+=1


\Else
{
	{\bf if} {marker[i] $>$ 0} {\bf then} marker[i] = -1	
	{\bf else} marker[i] -= 1
}

\If{marker[i] $<$ 0 or marker[i] $>=$ 10}
{
$TS[i] -= \alpha * BS$
}

}
}

\tcc{{\bf aggregation}}

exp\_TS = exp(TS)

$w = [0]*n$\tcp*{initial weight for each gradient}

\For{i in n}{
$w_i$ = exp\_TS[i] / sum(exp\_TS)
}


$g=\sum_{i=1}^n w_i * \text{G[i]}$


\Return $g$

\caption{{\sf RECESS} Defense: Proactive Detection and Aggregation}
\label{alg:recess}
\end{algorithm}

Compared with previous methods, {\sf RECESS} has four advantages:
\begin{enumerate}[(a)]
\item Provide a more accurate measure of the abnormality extent for each client's update to assign the trust score. 
\item Enable the defender to put clients’ previous performance into account, not only in each iteration, enhancing the detection fault tolerance.
\item Protect benign outliers while effectively detecting carefully-constructed poisoning gradient updates, improving model accuracy and generalization.
\item Outperform previous FLTrust and ML-based detections in terms of efficiency with no need for extra data.
\end{enumerate}

\subsection{Anomaly Detection Using the New Defined $Abnormality$}\label{sec:appdis}

The new defined \textit{Abnormality} can be also used in the defense direction of anomaly detection.
Using proposed proactive detection, the defender can calculate the abnormality value of each client to obtain an abnormality set, which can be easily divided into two distinct groups, with malicious clients in the smaller group. 
The details of the division algorithm are shown as follows.

When the defender receiving an abnormality set $A$, which includes the abnormality $\alpha$ of each client in the entire FL system, i.e., $A = \{\alpha_1, \alpha_2, ..., \alpha_n\}$, we can set a threshold $t$ to divide the abnormality set $A$ into two groups, which also means $A = A_b \cup A_m$ and $\emptyset = A_b \cap A_m$, where $A_b$ denotes the abnormality subset of all benign clients including outliers and $A_m$ denotes all malicious clients' ones.
This is based on our observation that there is an obvious numerical difference between the abnormality of malicious gradients and benign gradients by using our proactive detection.
Consequently, malicious clients are in the group with small size, i.e., $|A_m| < |A_b| $.

The dividing threshold $t$ can be found by the traversal of $A$. 
After sorting $A$ according to the abnormality value, we can find a value $v$ that minimizes the range overlap of the mean value between the two groups.
This value can be used as the dividing threshold, i.e., $t=v$.
The range can be plus or minus three times the standard deviation, i.e., $\{\bar{A_i}-3\sigma,\bar{A_i}+3\sigma\}$ where $A_i$ is any one of two groups and $\sigma$ is the corresponding standard deviation.
Besides, the binary search algorithm can also be used here to reduce computational complexity.

\subsection{Effectiveness Analysis}

The key point of {\sf RECESS} is to slightly perturb constructed aggregation gradients to all clients, and gradients from malicious clients solved based on other benign gradients generate larger variation, which is easily discriminated.
Intuitively, the accumulated increases of each benign gradient's variance result in a corresponding increase in the malicious gradient's variance. Then we show the malicious gradient varies more than the benign gradients.
We first consider the worst white-box case where the attacker knows other benign gradients.
\begin{proposition}\label{proposi}
For the model poisoning in an FL system, if the poisoning gradient (e.g., $g_1'$) is solved based on other benign gradients (e.g., $g_i$),
suppose the perturbation $\Delta$ added to each benign gradient, the variance of the poisoning gradient must be larger than the variance of any benign gradient, i.e.,
\begin{equation}
	Var(g_1' + \Delta_1) > \max_{c+1 \leq i \leq n} Var(g_i + \Delta_i)),
\end{equation}
where $i = c+1,...,n$.
\end{proposition}

\begin{proof}
The goal of the model poison attack is to let the attacked aggregation gradient $g'$ deviate as far away from the direction of the aggregation gradient $g$ without poisoning as possible.
When uploaded gradients are linearly independent, we formulate the model poison attack of each iteration into an optimization program based on \cite{fang2020local},
\begin{maxi}
{g_1',...,g_c'}{s^{\mathrm{T}} (g'-g)}{\label{optimization}}{}	
\addConstraint{g=}{Agg (g_1,..., g_c, g_{c+1},..., g_n)}
\addConstraint{g'=}{Agg (g_1',..., g_c', g_{c+1},..., g_n)}
\end{maxi}
where 
$g_i$ denotes the gradient update of $i$-th client to server, 
$g_i'$ denotes the poisoning gradient from client $i$ manipulated by the attacker, 
and $s^{\mathrm{T}}$ is a column vector of the changing directions of the aggregation gradient, the first $c$ clients are malicious, and the rest clients are benign.

Suppose the first client is malicious, and 
we have $g_{1}^{\prime}=g-\lambda s$, thus the optimization objective can be transformed into $s^{\mathrm{T}}(g'-g)+\lambda s^{\mathrm{T}} s$. 
Based on the literature \cite{fang2020local}, we can know that $s^{\mathrm{T}}(g-g')$ is a constant and $s^{\mathrm{T}} s = d$, where $d$ is the dimension number of the aggregation gradient, thus the optimization program (\ref{optimization}) can be converted into
\begin{maxi}
{}{\lambda}{}{}	
\addConstraint{g_1'}{=Agg(g_1',..., g_c', g_{c+1},..., g_n)}
\addConstraint{g_1'}{=g - \lambda s}.
\end{maxi}

This program is equal to the program (\ref{optimization}) and its solution is
\begin{equation}	
\begin{split}
\lambda = &\frac{1}{(n-2 c-1) \sqrt{d}} \cdot \min _{c+1 \leq i \leq n}\left(\sum_{l \in \tilde{\Gamma}_{g_{i}}^{n-c-2}} \|g_{l}, g_{i}\|_2 \right)\\
& +\frac{1}{\sqrt{d}} \cdot \max _{c+1 \leq i \leq n} \|g_{i}, g\|_2,
\label{solution}
\end{split}
\end{equation}
where $\tilde{\Gamma}_{g_{i}}^{n-c-2}$ denotes the set of $n-c-2$ benign gradients that are the closest to $g_i$ in $L_2$ distance.
Afterward, we abbreviate the solution (\ref{solution}) into 
\begin{equation}
	\lambda = D_1(g_l,g_i) + D_2(g_i,g), 
\end{equation}
where $D_1(g_l,g_i) = \frac{1}{(n-2 c-1) \sqrt{d}} \cdot \min _{c+1 \leq i \leq n}\left(\sum_{l \in \tilde{\Gamma}_{g_{i}}^{n-c-2}} \|g_{l}, g_{i}\|_2 \right)$ and $D_2(g_i,g) = \frac{1}{\sqrt{d}} \cdot \max _{c+1 \leq i \leq n} \|g_{i}, g\|_2$, so the gradient of the first malicious client should be
\begin{equation}
	g_1' = g - s D_1(g_l,g_i) - s D_2(g_i,g),
\end{equation}
and its variance can be calculated as:
\begin{equation}
	Var(g_1') = s^2 Var( D_1(g_l,g_i) +  D_2(g_i,g)).
\end{equation}

After adding the perturbation $\Delta_i$ on the gradients $g_i$, we can obtain 
\begin{equation}
\begin{split}
Var(g_1'+\Delta_1) = & s^2 Var(D_1(g_l,g_i + \Delta_i)) \\
& + s^2 Var(D_2(g_i + \Delta_i,g)).
\label{eq:var}
\end{split}	
\end{equation}

For the second term of equation (\ref{eq:var}), there is
\begin{equation}
\begin{aligned}
	s^2 Var(D_2(g_i + \Delta_i,g)) 
	&= s^2 Var(D_2(g_i + \Delta_i) \\
&= s^2 Var(\frac{1}{\sqrt{d}} \cdot \max _{c+1 \leq i \leq n} \|g_{i}+\Delta_i \|_2)\\
&=\max _{c+1 \leq i \leq n} Var(\|g_{i}+\Delta_i\|_2).
\end{aligned}
\end{equation}

As the first term of equation (\ref{eq:var}) is larger than zero, the following inequality holds
\begin{equation}
	Var(g_1'+\Delta_1) 
> \max _{c+1 \leq i \leq n} Var(\|g_{i}+\Delta_i\|_2).
\end{equation}
\end{proof}


For the black-box setting, the attacker leverages compromised clients' normal gradient updates without tampering to conjecture the aggregation direction to poison based on \cite{fang2020local,Shejwalkar2021ManipulatingTB}.
Those normal gradients have a similar distribution to the other benign gradients, thereby the property of black-box also satisfies the Proposition \ref{proposi}. 

\section{Evaluation}
\subsection{Setup}
\subsubsection{Datasets and FL Setting.}
The MNIST \cite{lecun1998gradient} is a simple handwritten digits dataset with low dimensions images and 10 balanced classes.
The CIFAR-10 \cite{krizhevsky2009learning} is a moderate dataset with 10 balanced classes of animals and vehicles. 
The Purchase \cite{c:2} is an unbalanced classification task.
The FEMNIST \cite{kairouz2021advances} is an explained character dataset including 52 upper and lower case letters and 10 digits from 3,400 clients.
The dataset is Non-IID and unbalanced, which is usually considered to be cross-device FL scenarios. 

Table \ref{tab:setting} shows four datasets and parameter settings used in the evaluation.
The IID and Non-IID are both considered in the dataset division.
We also consider two typical types of FL, i.e., cross-silo and cross-device.
In one FL iteration, each client updates its local model by one epoch. 

\begin{table*}[h]
\centering
\caption{Experiment datasets and FL settings.}
\label{tab:setting}
\resizebox{\linewidth}{!}{%
\begin{tabular}{cccccccccc}
\hline
Dataset  & Class & Size    & Dimension      & Model                             & Clients & Batch Size & Optimizer & Learning Rates & Epochs \\ \hline
MNIST    & 10    & 60,000  & $28 \times 28$ & FC ($784 \times 512 \times 10$)   & 100     & 100        & Adam      & 0.001          & 500    \\
CIFAR-10 & 10 & 60,000  & $32 \times 32$ & Alexnet \cite{krizhevsky2012imagenet} & 50 & 250              & SGD  & 0.5 to 0.05 at 1000th epoch & 1200 \\
Purchase & 100   & 197,324 & 600            & FC ($600 \times 1024 \times 100$) & 80      & 500        & SGD       & 0.5            & 1000   \\ 
FEMNIST  & 62 & 671,585 & $28 \times 28$ & LeNet-5 \cite{lecun1998gradient}      & $60 \subset 3400$ & client's dataset & Adam & 0.001               & 1500 \\
\hline
\end{tabular}%
}
\end{table*}

\begin{table*}[t]
\centering
\caption{Comparison results of FL accuracy with various defenses
against the latest poisoning attacks in both white-box and black-box cases.
In the white-box (black-box) case, the attacker has (no) knowledge of other benign gradients.
Each result is averaged over multiple repetitions.
}
\label{tab:results}
\resizebox{\linewidth}{!}{%
\begin{tabular}{ccccccccccc} 
\hline
\multirow{2}{*}{Dataset} & \multirow{2}{*}{\begin{tabular}[c]{@{}c@{}}Attacker's\\ Knowledge\end{tabular}} & \multirow{2}{*}{Attacks} & \multicolumn{8}{c}{Defences} \\ 
\cline{4-11}
 &  &  & Krum & Mkrum & Bulyan & Trmean & Median & FLTrust & DnC & RECESS \\ 
\hline
\multirow{11}{*}{\begin{tabular}[c]{@{}c@{}}MNIST\\0.962\\(FedAvg)\end{tabular}} &  & {\cellcolor[rgb]{0.875,0.875,0.875}}No Attack & {\cellcolor[rgb]{0.875,0.875,0.875}}0.889 & {\cellcolor[rgb]{0.875,0.875,0.875}}0.961 & {\cellcolor[rgb]{0.875,0.875,0.875}}0.954 & {\cellcolor[rgb]{0.875,0.875,0.875}}\textbf{0.962} & {\cellcolor[rgb]{0.875,0.875,0.875}}0.932 & {\cellcolor[rgb]{0.875,0.875,0.875}}0.958 & {\cellcolor[rgb]{0.875,0.875,0.875}}0.960 & {\cellcolor[rgb]{0.875,0.875,0.875}}0.959 \\ 
\cline{3-11}
 & \multirow{5}{*}{White-box} & LIE & 0.762 & 0.90 & 0.862 & 0.898 & 0.913 & \textbf{0.957} & 0.951 & 0.952 \\
 &  & Optimization attack & 0.681 & 0.846 & 0.875 & 0.944 & 0.915 & 0.902 & 0.946 & \textbf{0.951} \\
 &  & AGR-tailored & 0.547 & 0.775 & 0.869 & 0.852 & 0.888 & 0.893 & 0.937 & \textbf{0.948} \\
 &  & AGR-agnostic Min-Max & 0.875 & 0.798 & 0.896 & 0.846 & 0.896 & 0.923 & 0.950 & \textbf{0.957} \\
 &  & AGR-agnostic Min-Sum & 0.601 & 0.829 & 0.872 & 0.869 & 0.910 & 0.929 & 0.941 & \textbf{0.955} \\ 
\cline{2-11}
 & \multirow{5}{*}{Black-box} & LIE & 0.794 & 0.929 & 0.887 & 0.911 & 0.914 & 0.952 & \begin{tabular}[c]{@{}c@{}}\textbf{\textbf{0.960}}\\\end{tabular} & 0.959 \\
 &  & Optimization attack & 0.712 & 0.855 & 0.883 & 0.945 & 0.917 & 0.926 & 0.956 & \textbf{0.957} \\
 &  & AGR-tailored & 0.645 & 0.805 & 0.872 & 0.856 & 0.891 & 0.942 & 0.944 & \textbf{0.955} \\
 &  & AGR-agnostic Min-Max & 0.886 & 0.807 & 0.901 & 0.867 & 0.901 & 0.952 & 0.951 & 0.948 \\
 &  & AGR-agnostic Min-Sum & 0.633 & 0.835 & 0.877 & 0.877 & 0.912 & 0.951 & \textbf{0.957} & 0.954 \\ 
\hline

\multirow{11}{*}{\begin{tabular}[c]{@{}c@{}}CIFAR-10\\0.6605\\(FedAvg)\end{tabular}} &  & {\cellcolor[rgb]{0.875,0.875,0.875}}No Attack & {\cellcolor[rgb]{0.875,0.875,0.875}}0.5162 & {\cellcolor[rgb]{0.875,0.875,0.875}}0.6494 & {\cellcolor[rgb]{0.875,0.875,0.875}}0.6601 & {\cellcolor[rgb]{0.875,0.875,0.875}}\textbf{0.6605} & {\cellcolor[rgb]{0.875,0.875,0.875}}0.6582 & {\cellcolor[rgb]{0.875,0.875,0.875}}0.6341 & {\cellcolor[rgb]{0.875,0.875,0.875}}0.6409 & {\cellcolor[rgb]{0.875,0.875,0.875}}0.6554 \\ 
\cline{3-11}
 & \multirow{5}{*}{White-box} & LIE & 0.5058 & 0.5921 & 0.0965 & \textbf{0.6436} & 0.6243 & 0.6189 & 0.6047 & 0.6085 \\
 &  & Optimization attack & 0.3459 & 0.5817 & 0.6118 & 0.5066 & 0.5133 & 0.6075 & 0.5813 & \textbf{0.6206} \\
 &  & AGR-tailored & 0.2246 & 0.2915 & 0.2355 & 0.2276 & 0.2970 & 0.4636 & 0.4614 & \textbf{0.6043} \\
 &  & AGR-agnostic Min-Max & 0.5134 & 0.3096 & 0.3173 & 0.3709 & 0.2723 & 0.5229 & 0.5579 & \textbf{0.6202} \\
 &  & AGR-agnostic Min-Sum & 0.2248 & 0.4461 & 0.2211 & 0.4271 & 0.2765 & 0.5476 & 0.5347 & \textbf{0.6406} \\ 
\cline{2-11}
 & \multirow{5}{*}{\begin{tabular}[c]{@{}c@{}}Black-box\\\end{tabular}} & LIE & 0.3139 & 0.4512 & 0.3648 & 0.6063 & 0.5635 & \textbf{0.6319} & 0.6101 & 0.6318 \\
 &  & Optimization attack & 0.2891 & 0.5673 & 0.2605 & 0.5558 & ~0.5452 & 0.6284 & 0.6159 & \textbf{0.6390} \\
 &  & AGR-tailored & 0.2455 & 0.3123 & 0.2349 & 0.3155 & 0.2650 & 0.5534 & ~0.6066 & \textbf{0.6177} \\
 &  & AGR-agnostic Min-Max & 0.4078 & 0.3311 & 0.2948 & 0.2773 & 0.2934 & 0.5303 & 0.5771 & \textbf{0.6286} \\
 &  & AGR-agnostic Min-Sum & 0.2242 & 0.3435 & 0.2551 & 0.4381 & 0.3492 & 0.5094 & 0.5676 & \textbf{0.6432} \\ 
 \hline
 
\multirow{11}{*}{\begin{tabular}[c]{@{}c@{}}Purchase\\0.922\\(FedAvg)\end{tabular}} &  & {\cellcolor[rgb]{0.875,0.875,0.875}}No Attack & {\cellcolor[rgb]{0.875,0.875,0.875}}0.621 & {\cellcolor[rgb]{0.875,0.875,0.875}}0.919 & {\cellcolor[rgb]{0.875,0.875,0.875}}0.913 & {\cellcolor[rgb]{0.875,0.875,0.875}}\textbf{0.920} & {\cellcolor[rgb]{0.875,0.875,0.875}}0.874 & {\cellcolor[rgb]{0.875,0.875,0.875}}0.920 & {\cellcolor[rgb]{0.875,0.875,0.875}}0.919 & {\cellcolor[rgb]{0.875,0.875,0.875}}\textbf{0.920} \\ 
\cline{3-11}
 & \multirow{5}{*}{White-box} & LIE & 0.579 & 0.894 & 0.804 & 0.897 & 0.869 & \textbf{0.915} & 0.914 & 0.913 \\
 &  & Optimization attack & 0.561 & 0.782 & 0.676 & 0.902 & 0.872 & 0.905 & \textbf{0.912} & 0.911 \\
 &  & AGR-tailored & 0.008 & 0.705 & 0.626 & 0.686 & 0.676 & 0.845 & 0.884 & \textbf{0.911} \\
 &  & AGR-agnostic Min-Max & 0.015 & 0.715 & 0.619 & 0.575 & 0.728 & 0.891 & 0.897 & \textbf{0.914} \\
 &  & AGR-agnostic Min-Sum & 0.011 & 0.735 & 0.603 & 0.766 & 0.716 & 0.885 & 0.905 & \textbf{0.913} \\ 
\cline{2-11}
 & \multirow{5}{*}{Black-box} & LIE & 0.803 & 0.901 & 0.830 & 0.900 & 0.890 & 0.916 & 0.915 & \textbf{0.926} \\
 &  & Optimization attack & 0.577 & 0.797 & 0.704 & 0.904 & 0.884 & 0.911 & 0.918 & \textbf{0.921} \\
 &  & AGR-tailored & 0.013 & 0.737 & 0.629 & 0.698 & 0.733 & 0.891 & 0.895 & \textbf{0.914} \\
 &  & AGR-agnostic Min-Max & 0.019 & 0.720 & 0.630 & 0.650 & 0.740 & 0.908 & 0.906 & \textbf{0.920} \\
 &  & AGR-agnostic Min-Sum & 0.031 & 0.755 & 0.610 & 0.774 & 0.748 & 0.898 & 0.912 & \textbf{0.919} \\ 
\hline
\multirow{11}{*}{\begin{tabular}[c]{@{}c@{}}FEMNIST\\0.8235\\(FedAvg)\end{tabular}} &  & {\cellcolor[rgb]{0.875,0.875,0.875}}No Attack & {\cellcolor[rgb]{0.875,0.875,0.875}}0.6863 & {\cellcolor[rgb]{0.875,0.875,0.875}}0.8206 & {\cellcolor[rgb]{0.875,0.875,0.875}}0.7115 & {\cellcolor[rgb]{0.875,0.875,0.875}}0.8217 & {\cellcolor[rgb]{0.875,0.875,0.875}}0.7897 & {\cellcolor[rgb]{0.875,0.875,0.875}}\textbf{0.8237} & {\cellcolor[rgb]{0.875,0.875,0.875}}0.8201 & {\cellcolor[rgb]{0.875,0.875,0.875}}0.8230 \\ 
\cline{3-11}
 & \multirow{5}{*}{White-box} & LIE & 0.2687 & 0.5897 & 0.6514 & \textbf{0.8146} & 0.7828 & 0.8081 & 0.8048 & 0.7957 \\
 &  & Optimization attack & 0.2605 & 0.4981 & 0.6286 & 0.7920 & 0.7556 & 0.7848 & 0.7456 & \textbf{0.7983} \\
 &  & AGR-tailored & 0.5462 & 0.3696 & 0.4917 & 0.6127 & 0.4784 & 0.7615 & 0.3661 & \textbf{0.8145} \\
 &  & AGR-agnostic Min-Max & 0.6001 & 0.0606 & 0.6701 & 0.5717 & 0.5552 & 0.7845 & 0.6916 & \textbf{0.8019} \\
 &  & AGR-agnostic Min-Sum & 0.4750 & 0.2904 & 0.4496 & 0.5982 & 0.6264 & 0.7148 & 0.0694 & \textbf{0.7967} \\ 
\cline{2-11}
 & \multirow{5}{*}{Black-box} & LIE & 0.2221 & 0.8153 & 0.7045 & 0.8114 & 0.7151 & \textbf{0.8148} & 0.8145 & 0.8048 \\
 &  & Optimization attack & 0.5998 & 0.7411 & 0.6646 & 0.7893 & 0.6741 & 0.7807 & 0.7824 & \textbf{0.7937} \\
 &  & AGR-tailored & 0.5419 & 0.4009 & 0.4330 & 0.6858 & 0.4527 & 0.7651 & 0.7027 & \textbf{0.7745} \\
 &  & AGR-agnostic Min-Max & 0.3796 & 0.5631 & 0.6332 & 0.6440 & 0.6605 & 0.6848 & 0.7351 & \textbf{0.7850} \\
 &  & AGR-agnostic Min-Sum & 0.6135 & 0.7524 & 0.7052 & 0.6589 & 0.6785 & 0.7848 & 0.7122 & \textbf{0.7936} \\
\hline
\end{tabular}
}
\end{table*}

Eight defenses are considered in the comparison experiments.
Specifically, the common FedAvg is used as the benchmark.
Five classic robust rules,
Krum, Mkrum, Bulyan, Trmean, and Median, can provide convergence guarantees theoretically.
Two SOTA defenses, FLtrust and DnC, are also involved.
The parameters of these eight defenses are set to default.
For {\sf RECESS}, we set $\mathcal{A}=0.95$, $TS_0=1$, and $baseline\_decreased\_score=0.1$ unless otherwise specified. 

\subsubsection{Attack and Defense Settings.}

We consider the five latest poisoning attacks including LIE, optimization attack, AGR-tailored attack, AGR-agnostic Min-Max attack, and Min-Sum attack. 
We also set two scenarios, i.e., white-box and black-box, where attackers have and does not have knowledge of other benign gradients.
We assume 20\% malicious clients as the default value unless specified otherwise. 
For the test gradient construction, the modification allows arbitrary adjustments to gradient elements within a defined threshold. In our experiments, we select and add noise to the first 10\% of dimensions, iteratively adjusting until cosine similarity before/after meets the threshold.

The previous defenses are evaluated using default setting.
For Krum, we select the gradient with the smallest $L_2$ distance around $n-f-2$ gradients (e.g., $50-10-2=38$ in CIFAR). 
To ensure a fair evaluation, the parameter settings in our evaluation are aligned with previous works \cite{Cao2021FLTrustBF,Shejwalkar2021ManipulatingTB}.
The performance of baselines in our test is also consistent with the original claims. 

The detection window length depends on the setting of the trust score and the attack strength. Benefiting from our design, the setting of the detection threshold is fully determined by {\sf RECESS} itself through the weight of the aggregation, as shown in Equation \eqref{gradi} and \eqref{weight}.

We also consider the worst case where the attacker controls all malicious clients, i.e., all malicious clients colluded together, which is the same threat model as baselines. 
If the number of colluding clients decreases, the poisoned gradient will have a higher intensity in the aggregation, i.e., more abnormal, making it easier to be detected by {\sf RECESS}.


\subsubsection{Metric.}
$Accuracy$ is used to measure the performance of the well-trained model with or without poisoning.
The attack effect is measured by the drop in $accuracy$.
Higher $Accuracy$ indicates better defense. 


\subsubsection{Environment.}
All experiments are executed on an Ubuntu 16.04.1 LTS 64 bit server, where the CPU is Intel(R) Xeon(R) Silver 4214, memory is 128GB RAM, and GPU is one RTX 2080 Ti with 11GB dedicated memory. 
The programming language is Python 3.7.7 and the Pytorch framework is used to implement the machine learning models.

\subsection{Comparison with SOTA}
Table \ref{tab:results} shows the main comparison results. 
{\sf RECESS} obtains higher accuracy than two SOTA defenses under latest poisoning attacks.
Due to the space limit, here we mainly show the results on two datasets. 

\subsubsection{Defender's Goal (1): Defensive Effectiveness.} 
As we can see from each row in Table \ref{tab:results},
{\sf RECESS} can defend against model poisoning attacks, but previous defenses have limited effect,
especially for the strongest AGR attack series in the white-box case. 
The reasons are twofold:
\begin{enumerate}[(a)]
	\item Some malicious gradients evade defenses to be selected for aggregation.
	\item Benign outliers are misclassified as malicious gradients and removed from the aggregation, especially for Non-IID datasets (e.g., FEMNIST) with more outliers.  
\end{enumerate}
Moreover, benign outliers are intractable for all previous defenses.
FLTrust is better than other defenses since its normalization limits the large magnitude gradients.
In contrast, {\sf RECESS} can effectively distinguish benign outliers from malicious gradients, 
thus obtaining the highest accuracy, even in the most challenging white-box Non-IID case.


\begin{figure}[t]
	\centering
	\includegraphics[width=0.8\linewidth]{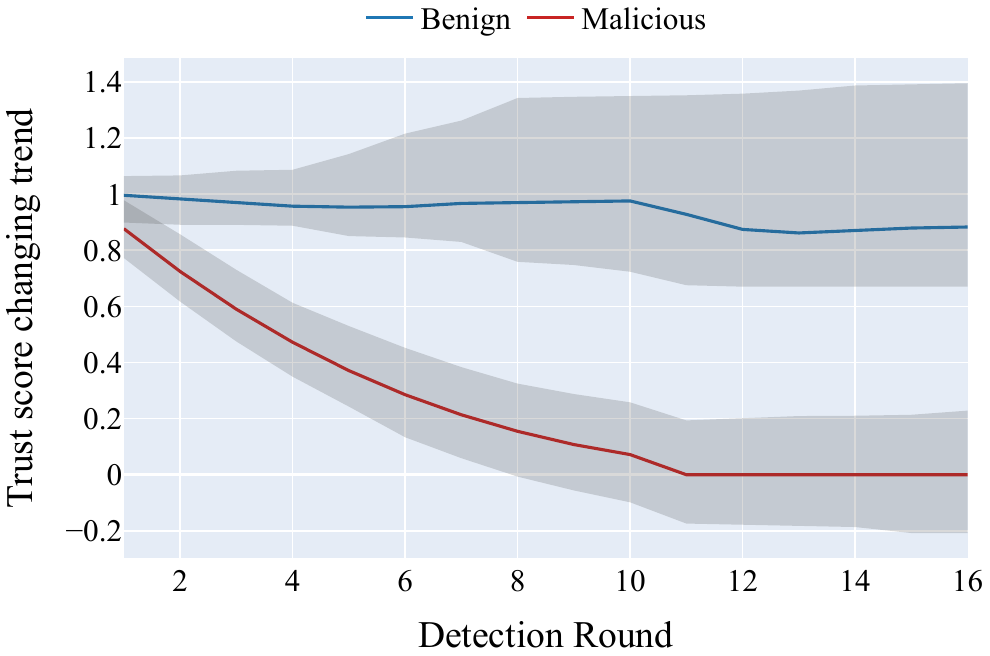}
	\caption{The changing trend of clients' trust score in {\sf RECESS}.
	The task is CIFAR-10.
	The model poisoning attack is the AGR-agnostic Min-Max attack.
	}
	\label{AGR_Score}
\end{figure}
\begin{figure}[t]
	\centering
	\includegraphics[width=\linewidth]{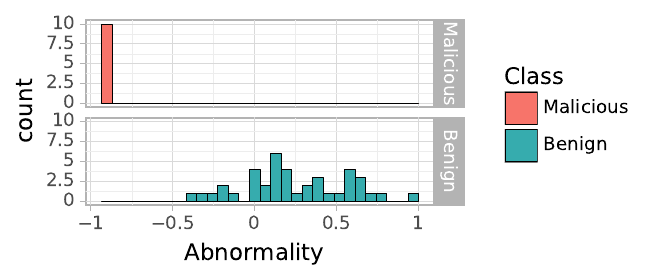}
	\caption{
	The histogram of clients' $Abnormality$ value  in {\sf RECESS}.
	The task is CIFAR-10.
	The model poisoning attack is the AGR-agnostic Min-Max attack.
	The detection round is \#3.
	}
	\label{AGR_dist}
\end{figure}

\subsubsection{Defender's Goal (2): Model Performance.} 
As shown in the first row in each set of Table \ref{tab:results} where
no attacks are under consideration,
robust rules will reduce the accuracy but have little impact on {\sf RECESS}.
This can be explained from three aspects:
\begin{enumerate}[(a)]
	\item Due to the existence of attackers, only 80\% benign clients participate in the aggregation compared with no attack (e.g. 40 vs 50 for CIFAR-10).
	\item Existing defenses discard benign gradients more or less, especially for benign outliers.
	\item Malicious gradients are selected for aggregation.
\end{enumerate}
The first aspect is inevitable since the defender cannot control the attacker, which also causes the main accuracy gap between FedAvg and our {\sf RECESS}.
However, {\sf RECESS} can improve the other two aspects 
by more accurately distinguishing whether gradients are malicious or benign.

\subsubsection{Impact of Attacker's Knowledge.}
Different knowledge affects the attack, and {\sf RECESS} outperforms other defenses in both white/black-box.
The stronger the attack, the better the effect of {\sf RECESS}.
In white-box cases, the attacker can elaborately construct poisoning gradients to maximize attack effectiveness with the knowledge of other benign gradients.
But in black-box cases, 
such adjustments are not possible without knowledge of other benign gradients, since 
the attacker cannot offset the effect of other benign gradients on the final aggregation, so the attack effect is greatly reduced.
This is the reason why defenses gain higher accuracy,
just look more effective, 
not themselves more effective.


\subsubsection{Changing Trend of Trust Score.}
Fig. \ref{AGR_Score} presents the trust score of benign and malicious clients in each round of {\sf RECESS} detection.
Here we execute the strongest AGR-agnostic Min-max attack.
Fig. \ref{opti_Score} illustrates the decreasing trend of trust scores for benign and malicious clients in {\sf RECESS} when dealing with the optimization attack.

Fig. \ref{AGR_dist} shows the histogram of each client's $Abnormality$ value in the third detection round to illustrate the effectiveness of {\sf RECESS}.
Fig. \ref{opti_dist} also shows the representative histogram of each client's $Abnormality$ value in the seventh detection round.
{\sf RECESS} can effectively distinguish malicious and benign users in both the optimization attack and AGR attacks.




\subsubsection{Comparison with FLDetector.}

The recent work FLDetector \cite{zhang2022fldetector} also adopts a "trust scoring" approach to identify malicious clients. The key difference lies in how anomalies are detected. We actively send test queries and compare the updates returned from clients, while FLDetector makes predictions based on historical updates. We also evaluated FLDetector, with the following results shown in Table \ref{tab:fldetector}.

\begin{table}[h]
\caption{Comparison results between RECESS and FLDetector under different model poisoning attacks.}
\label{tab:fldetector}
\resizebox{\columnwidth}{!}{%
\begin{tabular}{@{}cccc@{}}
\toprule
\textbf{Dataset}                        & \textbf{Attacks}     & \textbf{FLDetector} & \textbf{RECESS} \\ \midrule
\multirow{2}{*}{CIFAR-10 (IID)}         & Optimization attack  & 0.6384              & \textbf{0.6390} \\
                                        & AGR-agnostic Min-Max & 0.5178              & \textbf{0.6286} \\ \cmidrule(l){2-4} 
\multirow{2}{*}{CIFAR-10 (Non-IID 0.5)} & Optimization attack  & 0.6041              & \textbf{0.6145} \\
                                        & AGR-agnostic Min-Max & 0.5514              & \textbf{0.5912} \\ \cmidrule(l){2-4} 
\multirow{2}{*}{CIFAR-10 (Non-IID 0.8)} & Optimization attack  & 0.3544              & \textbf{0.6018} \\
                                        & AGR-agnostic Min-Max & 0.0176              & \textbf{0.5841} \\ \cmidrule(l){2-4} 
\multirow{2}{*}{FEMNIST}                & Optimization attack  & 0.7217              & \textbf{0.7937} \\
                                        & AGR-agnostic Min-Max & 0.5616              & \textbf{0.7850} \\ \bottomrule
\end{tabular}%
}
\end{table}

\begin{figure}
	\centering
	\includegraphics[width=0.8\linewidth]{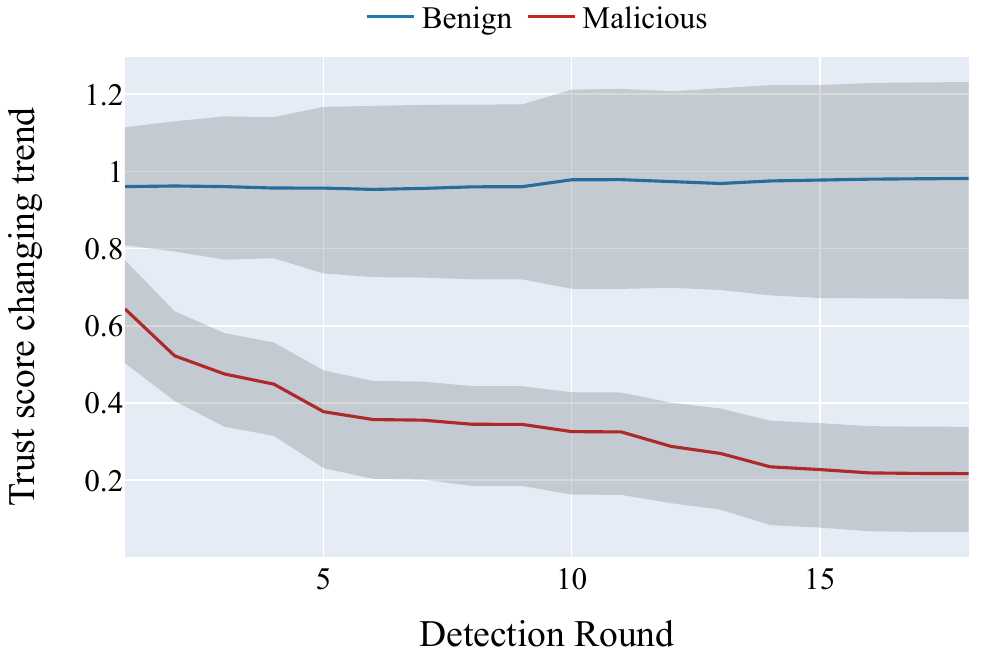}
	\caption{The changing trend of clients' trust score in {\sf RECESS}.
	The task is CIFAR-10.
	The model poisoning attack is the optimization attack.
	}
	\label{opti_Score}
\end{figure}

\begin{figure}
	\centering
	\includegraphics[width=\linewidth]{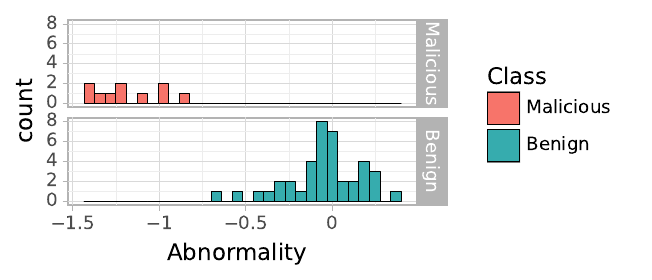}
	\caption{
	The histogram of clients' $Abnormality$ value in {\sf RECESS}.
	The task is CIFAR-10.
	The model poisoning attack is the optimization attack.
	The detection round is \#7.
	}
	\label{opti_dist}
\end{figure}

\subsection{Impact of FL Setting}
In this part, we evaluate {\sf RECESS} under various FL settings 
including the distribution of clients' datasets, the number of malicious clients, and the FL cases, 
to illustrate the practicability of our proposed scheme.

\begin{figure}[t]
	\centering
	\includegraphics[width=\linewidth]{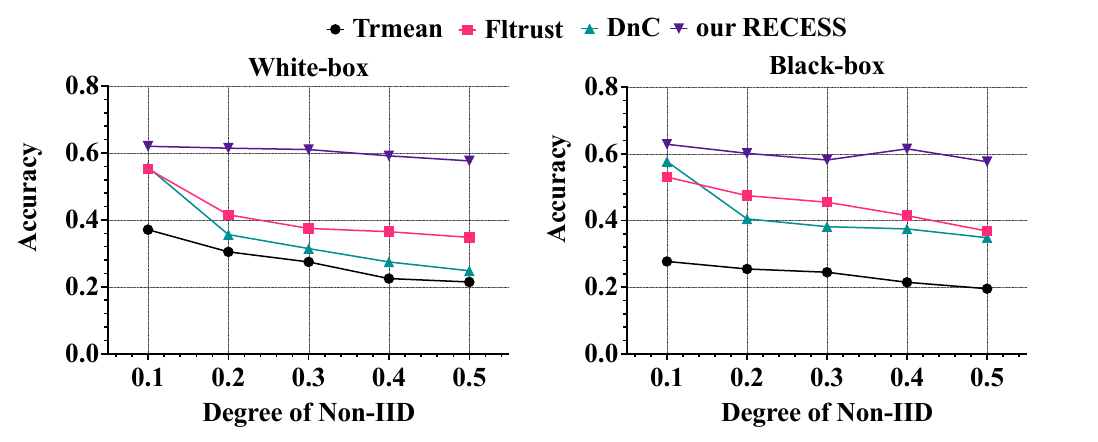}
	\caption{The impact of the local dataset's Non-IID degree on the FL defenses.
	The task is CIFAR-10.
	The model poisoning attack is AGR Min-Max.
	In the white-box (black-box) case, the attacker has (no) knowledge of other benign gradients.
	Our {\sf RECESS} achieves the best and most stable defense effect.
	}
	\label{Degree}
\end{figure}

\begin{figure}[t]
	\centering
	\includegraphics[width=\linewidth]{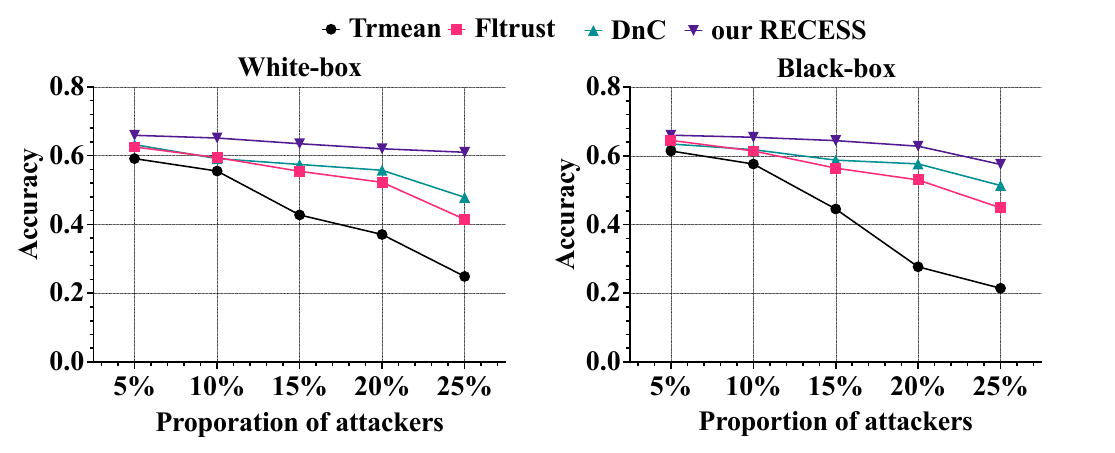}
	\caption{The impact of the proportion of attackers on the FL defenses.
	The task is CIFAR-10.
	The model poisoning attack is AGR Min-Max.
	In the white-box (black-box) case, the attacker has (no) knowledge of other benign gradients.
	As the proportion increases, the accuracy of using our {\sf RECESS} decreases minimally.
	}
	\label{Number}
\end{figure}


\subsubsection{Non-IID Degree of Dataset (More Outliers).}\label{exp:31}
We evaluate {\sf RECESS} on different Non-IID degrees of datasets.
The results on Non-IID FEMNIST in Table \ref{tab:results} show the effectiveness of our {\sf RECESS} against latest poisoning attacks.
Besides, we also involve the CIFAR-10 in the evaluation using the Non-IID dataset construction method same as \cite{Cao2021FLTrustBF,Shejwalkar2021ManipulatingTB}.
With $M$ classes, clients split into $M$ groups. Class $i$ example assigned to group $i$ with probability $q$, and to any group with probability $(1-q)/(M-1)$. $q=1/M$ is IID, larger $q$ means more non-IID.
We consider four defenses (the best classic robust rules Trmean, SOTA FLtrust and DnC, and {\sf RECESS}) against the strongest AGR-agnostic Min-Max attack in white/black-box cases.

Fig. \ref{Degree} shows that the increase in the Non-IID degree leads to an expansion of benign outliers which are similar to malicious gradients,
, which greatly increases the difficulty of previous defenses.
As the Non-IID degree increases, outliers increase and the detection becomes more difficult,
and previous defenses cannot effectively distinguish outliers from malicious gradients. 
However, it has little impact on {\sf RECESS}, since our method can accurately detect malicious gradients and benign outliers.



\subsubsection{Number of Malicious Clients.}

Fig. \ref{Number} shows that
RECESS remains an outstanding defending effect all along, as the number of malicious clients increases,
all poisoning attacks are more powerful, 
while the defending effect of other defenses decreases sharply.
We vary the proportion of malicious clients from 5\% to 25\%, consistent with \cite{Shejwalkar2021ManipulatingTB}.
The other settings remain the same as Section \ref{exp:31}.


\subsubsection{Cross-device FL.}
Previous evaluations are mostly under cross-silo settings except for the FEMNIST which is naturally cross-device shown in Table \ref{tab:results}.
Further, we consider the cross-device setting using the dataset CIFAR-10. 
In each epoch, the server stochastically selects 10 updates from all clients to aggregate.
The attack steps remain the same, and the other settings are the same as the default.

Table \ref{FLcase} shows that similar to the result of cross-silo, {\sf RECESS} still achieves the best defense effect.
Besides, the poisoning impacts under the cross-device setting are lower than the one of the cross-silo setting, because the server selects less number of clients for aggregation and ignores more malicious gradients in cross-device FL, thus the attacker cannot continuously poison, weakening the impact of the poisoning.
Consequently, this also leaves less improvement space for our {\sf RECESS} and other defenses.

\subsubsection{Deeper Models.}
To demonstrate scalability, we conducted experiments on larger models including DNN, VGG, and ResNet. As shown in Table \ref{tab:deepmodel}, our method achieves consistent utility for these complex models.
These results highlight the general applicability of our framework across model depths.

\begin{table}
\caption{Results on more complex models.}
\label{tab:deepmodel}
\resizebox{\columnwidth}{!}{%
\begin{tabular}{ccccc}
\hline
\textbf{Dataset}          & \textbf{Model}            & \textbf{Attacks} & \textbf{FedAvg} & \textbf{RECESS} \\ \hline
\multirow{2}{*}{MNIST}    & \multirow{2}{*}{DNN}      & No Attack        & \textbf{0.9487} & 0.9405          \\
 &  & Optimization attack & 0.6873 & \textbf{0.9314} \\ \cline{2-5} 
\multirow{6}{*}{CIFAR-10} & \multirow{3}{*}{ResNet20} & No Attack        & \textbf{0.8449} & 0.8217          \\
 &  & Optimization attack & 0.1718 & \textbf{0.8173} \\
 &  & AGR-tailored Attack & 0.1544 & \textbf{0.8014} \\ \cline{3-5} 
                          & \multirow{3}{*}{VGG11}    & No Attack        & \textbf{0.7515} & 0.7408          \\
 &  & Optimization attack & 0.5032 & \textbf{0.7344} \\
 &  & AGR-tailored Attack & 0.3834 & \textbf{0.7037} \\ \hline
\end{tabular}%
}
\end{table}

\subsection{Defensive Effectiveness of Other Poisoning Attacks}
Though this work focuses primarily on untargeted model poisoning, along with previous SOTA works,
{\sf RECESS} detects inconsistency between client behaviors, so it could mitigate targeted/untargeted attacks from clients with backdoored/poisoned data by replacing weighted averaging with a Byzantine-robust aggregation like Median. 

To demonstrate {\sf RECESS}'s versatility, we evaluated various attacks, including data poisoning (label flipping) and targeted backdoor (scaling attack).
The results in Table \ref{tab:backdoor} show that {\sf RECESS} with Median effectively defends against the label flipping attack and scaling attack.
{\sf RECESS} with Median also achieves comparable accuracy to the normal model under no attack and resists the backdoor.

\begin{table*}
\centering
\caption{FL accuracy with RECESS against other types of poisoning attacks.}
\label{tab:backdoor}
\resizebox{0.85\linewidth}{!}{%
\begin{tabular}{cccccccc}
\hline
\textbf{Datasets} &
  \textbf{Attacks} &
  \textbf{Metrics} &
  \textbf{FedAvg} &
  \textbf{FLTrust} &
  \textbf{DnC} &
  \textbf{RECESS} &
  \textbf{\begin{tabular}[c]{@{}c@{}}RECESS with \\ Median\end{tabular}} \\ \hline
\multirow{4}{*}{MNIST}    & No Attack                       & Model Accuracy      & \textbf{0.9621} & 0.9584       & 0.9601 & 0.9598 & 0.9581          \\
                          & Label Flipping                  & Model Accuracy      & 0.9018          & 0.9448       & 0.9225 & 0.9154 & \textbf{0.9548} \\
                          & \multirow{2}{*}{Scaling Attack} & Model Accuracy      & \textbf{0.9542} & 0.9428       & 0.9518 & 0.9539 & 0.9519          \\
                          &                                 & Attack Success Rate & 1.0             & \textbf{0.0} & 1.0    & 1.0    & \textbf{0.0}    \\ \cline{2-8} 
\multirow{4}{*}{CIFAR-10} & No Attack                       & Model Accuracy      & \textbf{0.6605} & 0.6341       & 0.6409 & 0.6554 & 0.6383          \\
                          & Label Flipping                  & Model Accuracy      & 0.4145          & 0.5748       & 0.4284 & 0.4415 & \textbf{0.6148} \\
                          & \multirow{2}{*}{Scaling Attack} & Model Accuracy      & \textbf{0.6329} & 0.6168       & 0.6123 & 0.6314 & 0.6228          \\
                          &                                 & Attack Success Rate & 1.0             & 0.0248       & 0.6412 & 1.0    & \textbf{0.0}    \\ \hline
\end{tabular}%
}
\end{table*}

The setting of the label flipping attack is the same as \cite{fang2020local} and the scaling’s setting is consistent with \cite{bagdasaryan2020backdoor}. The accuracy of the model using {\sf RECESS} with Median under all attacks rivals conventional models (using FedAvg under no attack). This shows {\sf RECESS} can effectively defend against diverse FL poisoning. While not demonstrated originally, we respectfully suggest our technique is capable of tackling these attacks straightforwardly.

\subsection{Sensitivity of Hyperparameters}\label{sec:sensitivity}
The setting of the detection frequency needs to consider the specific task and attack strength. 
In our tests, we set $TS_0 = 1$ and $baseline\_decreased\_score = 0.1$, considering that performing $10$ detections is sufficient to confirm the malicious updates. 
In particular, the stronger the attack is, the easier it is to be detected in early detection, resulting in a particularly fast decline in its trust score and thus less impact on aggregation. 
Here we perform more fine-grained analysis of the hyperparameters setting.

\subsubsection{Deduction Speed of Trust Score.}
The detection speed is controlled by the initial trust score and the baseline decreased score each round.
To analyze the effect of the rate of trust score decline on the recess mechanism, given an initial trust score of $1.0$, i.e., $TS_0=1.0$, we vary the $baseline\_decreased\_score$ to $0.05$, $0.1$, and $0.2$ to see the performance of recess detection.

Fig. \ref{fig:changepara} shows that the larger $baseline\_decreased\_score$, the fewer rounds required for detection, and the faster the detection speed. However, when the $baseline\_decreased\_score$ is set large, the client's trust score drops too fast, which also leads to benign users being easily affected, so in this work when the initial score is $1.0$, we choose $baseline\_decreased\_score=0.1$, which takes a tradeoff in terms of detection speed and stability.	

\begin{figure*}[h]
	\centering
	\begin{subfigure}[t]{0.32\textwidth}
		\centering
		\includegraphics[width=\textwidth]{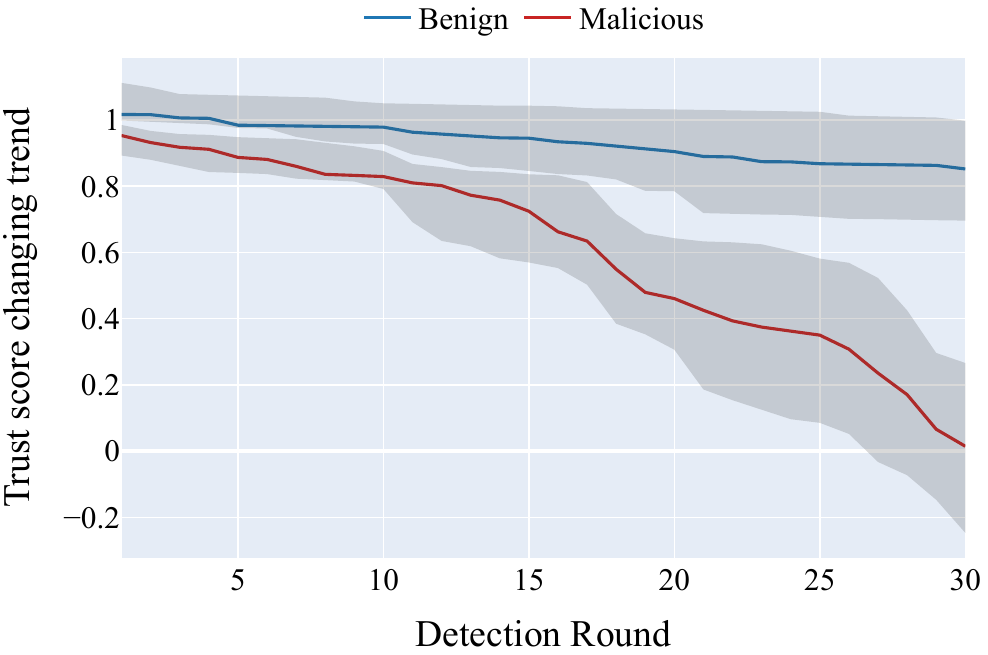}
		\caption{$baseline\_decreased\_score = 0.05$}
	\end{subfigure}
	\hfill
	\begin{subfigure}[t]{0.32\textwidth}
		\centering
		\includegraphics[width=\textwidth]{fig/score/AGR_Score_changing}
		\caption{$baseline\_decreased\_score = 0.1$}
	\end{subfigure}
	\hfill
	\begin{subfigure}[t]{0.32\textwidth}
		\centering
		\includegraphics[width=\textwidth]{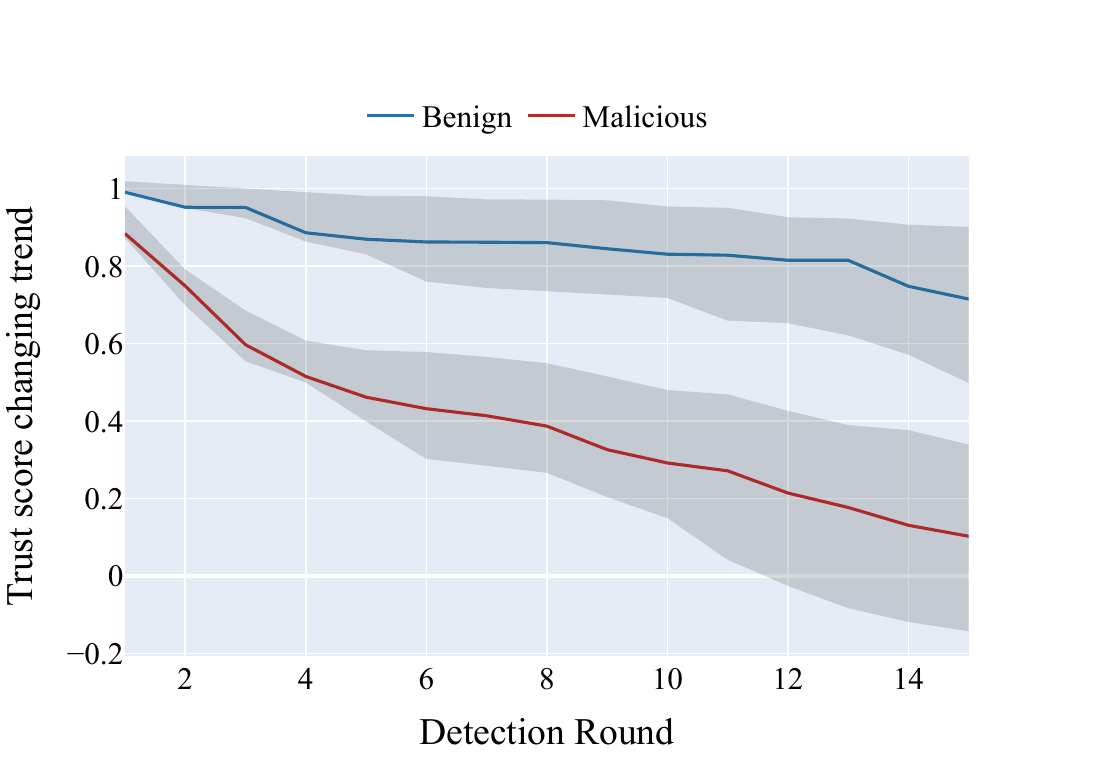}
		\caption{$baseline\_decreased\_score = 0.2$}
	\end{subfigure}
	\caption{
	The comparison of trust score changing trend with varying $baseline\_decreased\_score$.
	We set the initial trust score to $1.0$.
	Each result is averaged over multiple repetitions.
	}
	\label{fig:changepara}
\end{figure*}

\subsubsection{Test Gradient Direction Adjustment Threshold.}\label{sec:threshold}
We use the cosine similarity to measure the change of the test gradient direction. Note that a larger similarity represents a smaller amount of adjustment in direction, e.g. the four cases of $<5$, $<15$, $<25$, and $<35$ degrees on the x-axis correspond to cosine similarities of $>0.9961$, $>0.9659$, $>0.9063$, and $>0.8191$, respectively.
We test the detection effect of {\sf RECESS} under different thresholds, including detection accuracy and concealment, respectively.
Here we directly determine the detection accuracy by predicting the client as benign or malicious, i.e., an client is considered an attacker if the value of $Abnormality$ is positive.
We also let the attacker perform the adaptive attack shown in Section \ref{sec:adaptive}, i.e., to check whether the received gradient is the test gradient, and the attacker's judgment success rate $A_S$ is used as the concealment indicator, i.e.,
$concealment = 1 - 2 \times |A_S - 0.5|$, since the attacker is blind guessing when $A_S$ is close to 0.5.

Fig. \ref{fig:threshold_sensitivity} shows the results. The smaller the perturbation threshold, the higher the detection accuracy, which means that the poisoning attack is easier to be detected.
This is because according to the attacker's optimization equation, the smaller the variation in the received aggregation gradient, the attacker tends to add more perturbations and thus is more likely to be detected. 
However, test gradients with small variations are easily screened by the attacker and thus are adaptively evaded, so we need to trade-off detection accuracy and concealment. In the main experiment we chose 18 degrees as the threshold $\mathcal{A}$, i.e., 0.9510, which is reasonable in practice.

\begin{figure}[h]
	\centering
	\includegraphics[width=0.8\linewidth]{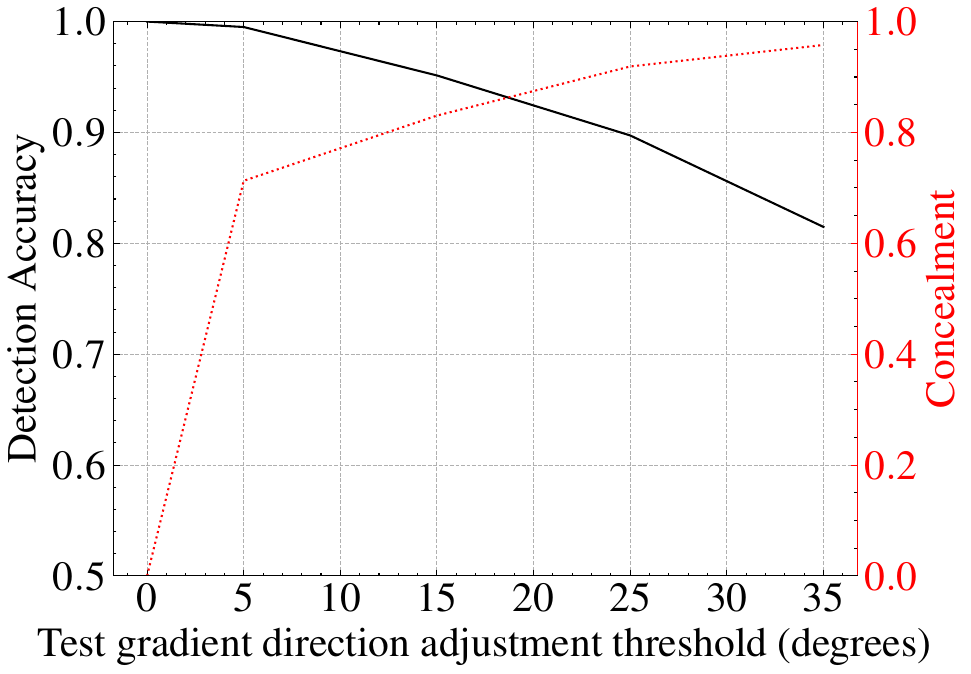}
	\caption{The results of threshold $\mathcal{A}$ sensitivity analysis.}
	\label{fig:threshold_sensitivity}
\end{figure}

%
%

\subsection{Adaptive Attacks}\label{sec:adaptive}
\subsubsection{Active Evasion.}
With knowledge of RECESS, the attacker checks if it's a test gradient before deciding to poison or not. The adaptive strategy involves estimating an aggregated gradient $g_p$ in each iteration from controlled benign clients. Then, the attacker compares the received aggregated gradient $g_{agg}$ with $g_p$ and $g_i^{k-1}$ from the last iteration. If cosine similarity $S_C (g_{agg}, g_i^{k-1}) \geq S_C (g_{agg}, g_p)$, the poisoning begins, while if $S_C (g_{agg}, g_i^{k-1}) < S_C (g_{agg}, g_p)$, the poisoning stops, as it's considered detected.

As a countermeasure, the defender boosts detection frequency and intersperses it during the entire training, instead of fixed consecutive detection, and increases the sensitivity of parameters. 
Note that the defense should aim to make the model converge, it is not necessary to find the attacker. In other words, {\sf RECESS}'s presence increases the attack cost, so even if the attacker evades detection, the loss of model accuracy is negligible, which is acceptable.

Table \ref{fig:adaptive} shows that the adaptive attack has little impact on the final model accuracy though it is stealthier than the continuous attack. This is due to the similarity between the test gradient and the normal aggregated gradient, leading to inaccurate estimation by the attacker and a decrease in aggregation weight. Also, increasing detection frequency lowers poisoning frequency, so the poisoning impact will be gradually alleviated by normal training over time.

\begin{table*}
\centering
\caption{
The impact of cross-device FL on the defenses.
The model poisoning attack is AGR Min-Max.
}
\label{FLcase}
\resizebox{0.8\linewidth}{!}{%
\begin{tabular}{ccccccc} 
\hline
\multirow{2}{*}{Dataset} & \multirow{2}{*}{\begin{tabular}[c]{@{}c@{}}Attacker's\\ Knowledge\end{tabular}} & \multirow{2}{*}{Attacks} & \multicolumn{4}{c}{Defences} \\ 
\cline{4-7}
 &  &  & Trmean & FLTrust & DnC & RECESS \\ 
\hline
\multirow{11}{*}{CIFAR-10} & {\cellcolor[rgb]{0.875,0.875,0.875}} & {\cellcolor[rgb]{0.875,0.875,0.875}}No Attack & {\cellcolor[rgb]{0.875,0.875,0.875}}\textbf{0.6490} & {\cellcolor[rgb]{0.875,0.875,0.875}}0.6178 & {\cellcolor[rgb]{0.875,0.875,0.875}}0.6284 & {\cellcolor[rgb]{0.875,0.875,0.875}}0.6448 \\ 
\cline{3-7}
 & \multirow{5}{*}{White-box} & LIE & 0.5559 & 0.6101 & 0.6157 & \textbf{0.6228} \\
 &  & Optimization attack & 0.5655 & 0.6171 & 0.5913 & \textbf{0.6439} \\
 &  & AGR-tailored & 0.5080 & 0.4538 & 0.5915 & \textbf{0.6046} \\
 &  & AGR-agnostic Min-Max & 0.5548 & 0.5226 & 0.5478 & \textbf{0.6379} \\
 &  & AGR-agnostic Min-Sum & 0.5758 & 0.5964 & 0.5793 & \textbf{0.6393} \\ 
\cline{2-7}
 & \multirow{5}{*}{Black-box} & LIE & 0.5999 & 0.6137 & 0.6128 & \textbf{0.6353} \\
 &  & Optimization attack & 0.6168 & 0.6159 & 0.6209 & \textbf{0.6277} \\
 &  & AGR-tailored & 0.5509 & 0.5485 & ~0.5956 & \textbf{0.6142} \\
 &  & AGR-agnostic Min-Max & 0.5833 & 0.5715 & 0.6095 & \textbf{0.6416} \\
 &  & AGR-agnostic Min-Sum & 0.6068 & 0.5937 & 0.6074 & \textbf{0.6385} \\
\hline
\end{tabular}
}
\end{table*}

\begin{table*}
\centering
\caption{FL accuracy with enhanced {\sf RECESSS} defenses against the adaptive poisoning attack. The task is FEMNIST. The detection frequency is contolled by $TS_0$ and $baseline\_decreased\_score$.}
\label{fig:adaptive}
\resizebox{0.8\linewidth}{!}{%
\begin{tabular}{ccccccc} 
\hline
\multirow{2}{*}{\begin{tabular}[c]{@{}c@{}}Attacker's \\Knowledge\end{tabular}} & \multirow{2}{*}{Attacks} & \multicolumn{5}{c}{RECESS Detection Frequency} \\ 
\cline{3-7}
 &  & 0 & 10 & 50 & 100 & 200 \\ 
\hline
\multirow{3}{*}{white-box} & {\cellcolor[rgb]{0.753,0.753,0.753}}no attack & {\cellcolor[rgb]{0.753,0.753,0.753}}0.8235 & {\cellcolor[rgb]{0.753,0.753,0.753}}0.823 & {\cellcolor[rgb]{0.753,0.753,0.753}}0.8221 & {\cellcolor[rgb]{0.753,0.753,0.753}}0.8224 & {\cellcolor[rgb]{0.753,0.753,0.753}}0.8215 \\
 & optimization attack & 0.2182 & 0.7983 & 0.8048 & 0.8106 & 0.8114 \\
 & adaptive optimization attack & 0.2482 & 0.3456 & 0.6847 & 0.7815 & 0.8011 \\ 
\hline
\multirow{2}{*}{black-box} & optimization attack & 0.5152 & 0.7937 & 0.8045 & 0.8117 & 0.8048 \\
 & adaptive optimization attack & 0.5248 & 0.6481 & 0.7847 & 0.8048 & 0.8148 \\
\hline
\end{tabular}
}
\end{table*}

\subsubsection{Poisoning before the Algorithm Initialization.}
The detection basis is that the variance of the malicious gradient is larger than the benign gradient. Malicious gradients, solved by the optimization problem, are more inconsistent. 
It is also theoretically proven that the optimization problem amplifies this variance. Thus, {\sf RECESS} is unaffected by initial poisoned gradients as it does not change this basis. 
Besides, when the poisoned aggregated gradient was used to detect, most clients will be identified as malicious.
This violates the assumption of Byzantine robustness that requires more than 51\% of users to be benign. Hence, the defender can easily discern a potentially poisoned initial gradient based on the new abnormality definition of {\sf RECESS}.

\subsubsection{Inconsistent attacks.}
Attackers can also first pretend as benign client to increase the trust score, and then provide the poisoning gradients. However, {\sf RECESS} detects such behavior inconsistencies over time. The trust scoring of {\sf RECESS} also incorporates delayed penalties for discrepancies between a client's current and past behaviors (shown in Section \ref{sec:scheme}). Additionally, three factors limit intermittent attacks' impact:
\begin{enumerate}[(a)]
	\item FL's self-correcting property means inconsistent poisoning is much less impactful. Attackers would not take this approach in practice.
	\item In real settings, clients participate briefly, often just once. Attackers cannot afford to waste rounds acting benign before attacking.
	\item Defenses aim to accurately detect malicious updates for model convergence. Even if poisoning temporarily evaded detection, attack efficacy would diminish significantly, making it no longer a serious security concern.
\end{enumerate}

\section{Discussion}



\textbf{Cost.}
As the first proactive detection, {\sf RECESS} is expected to be deployed on the server-side without the overhead of local clients, and the main cost is the number of iterations used for detection because our proactive detection needs clients to cooperate, thus consuming certain iterations.
For $r$ iterations of detection, the communication and computational cost are $r$ and $n*r$ respectively, where $n$ is the number of clients with customized test gradients. In practice, we find that $r$ is a small number thanks to our effective detection, e.g., $10$ is sufficient for most cases, which is negligible for the overall convergence (usually $500-1500$).

\textbf{Scalability.}
{\sf RECESS} can adapt to different FL setting in which the server gets the local model updates and then aggregate them directly.    
This adaptation can be achieved by converting uploaded models into gradients, since sending model updates is considered equivalent to sending gradients.

Besides, it is an arms race to develop defenses against new emerging attacks including backdoor. 
It can be considered from two perspectives: stealthiness and attack effect, so the corresponding defenses need to consider both detection and aggregation \cite{gong2022backdoor}. 
Thus, RECESS is a promising method of detecting new attacks including the backdoor since it has been improved on both proactive detection and robustness aggregation. For more stealthy attacks, RECESS can lengthen the detection window and also make the detection more sensitive to gradient changes by adjusting the magnitude of test gradients. For attacks with better effect, RECESS can cope with the latest attack, and the robustness aggregation can effectively limit the effect of malicious gradients on aggregation.



\section{Conclusion}

In this work, we proposed {\sf RECESS}, a novel defense 
for FL against the latest model poisoning attacks.
We shifted the classic reactive analysis to proactive detection and 
offered a more robust aggregation mechanism for the server.
In the comprehensive evaluation, {\sf RECESS} achieves better performance than SOTA defenses and solves the outlier detection problem that previous methods can not handle.
We anticipate that our {\sf RECESS} will provide a new research direction for poisoning attack defense and promote the application of highly robust FL in practice.

\bibliographystyle{IEEEtran} 
\bibliography{PoisonFL}

\balance

\end{document}